\documentclass[11pt]{article}
\usepackage{amssymb, amsmath}
\usepackage{fullpage,subfigure}

\usepackage{graphicx,amsmath,amsfonts,amscd,amssymb,bm,cite,epsfig,epsf,url,color,algorithm,algorithmic}
\usepackage{fullpage} \usepackage[small,bf]{caption}
\newtheorem{theorem}{Theorem}[section]
\newtheorem{lemma}[theorem]{Lemma}

\newcommand{\R}{\mathbb{R}}
\newcommand{\C}{\mathbb{C}}

\newcommand{\<}{\langle}
\renewcommand{\>}{\rangle}

\newcommand{\goto}{\rightarrow}

\renewcommand{\P}{\operatorname{\mathbb{P}}}
\newcommand{\E}{\operatorname{\mathbb{E}}}

\newcommand{\vct}[1]{\bm{#1}}
\newcommand{\mtx}[1]{\bm{#1}}

\newcommand{\bx}{\bm{x}}
\newcommand{\bX}{\bm{X}}
\newcommand{\bH}{\bm{H}}

\newcommand{\bI}{\bm{I}}
\newcommand{\bU}{\bm{U}}
\newcommand{\bY}{\bm{Y}}
\newcommand{\bZ}{\bm{Z}}
\newcommand{\bb}{\bm{b}}
\newcommand{\by}{\bm{y}}
\newcommand{\bu}{\bm{u}}
\newcommand{\bv}{\bm{v}}
\newcommand{\be}{\bm{e}}
\newcommand{\bz}{\bm{z}}

\newcommand{\rank}{\operatorname{rank}}

\newcommand{\tr}{\operatorname{Tr}}

\numberwithin{equation}{section}

\def \endprf{\hfill {\vrule height6pt width6pt depth0pt}\medskip}
\newenvironment{proof}{\noindent {\bf Proof} }{\endprf\par}

\newcommand{\cA}{\mathcal{A}}

\author{Emmanuel J. Cand\`{e}s\thanks{Departments of Mathematics and
    of Statistics, Stanford University, Stanford CA 94305}, Thomas
  Strohmer\thanks{Department of Mathematics, University of California
    at Davis, Davis CA}\,\, and Vladislav Voroninski\thanks{Department
    of Mathematics, University of California at Berkeley, Berkeley
    CA}}

\title{PhaseLift: Exact and Stable Signal Recovery from Magnitude
  Measurements via Convex Programming}

\date{September 2011}

\begin{document}

\maketitle

\begin{abstract}
  Suppose we wish to recover a signal $\bx \in \C^n$ from $m$
  intensity measurements of the form $|\langle \bx,\bz_i \rangle|^2$,
  $i = 1, 2, \ldots, m$; that is, from data in which phase information
  is missing. We prove that if the vectors $\bz_i$ are sampled
  independently and uniformly at random on the unit sphere, then the
  signal $\bx$ can be recovered exactly (up to a global phase factor)
  by solving a convenient semidefinite program---a trace-norm
  minimization problem; this holds with large probability provided
  that $m$ is on the order of $n \log n$, and without any assumption
  about the signal whatsoever. This novel result demonstrates that in
  some instances, the combinatorial phase retrieval problem can be
  solved by convex programming techniques. Finally, we also prove 
  that our methodology is robust vis a vis additive noise.
\end{abstract}

\section{Introduction}
\label{sec:intro}

In many applications, one would like to acquire information about an
object but it is impossible or very difficult to measure and record
the phase of the signal. The problem is then to reconstruct the object 
from intensity measurements only. A problem of this kind that has
attracted a considerable amount of attention over the last hundred
years or so, is of course that of recovering a signal or image from the
intensity measurements of its Fourier transform~\cite{Hur89,KST95} as
in X-ray crystallography. As is well-known, such phase retrieval
problems are notoriously difficult to solve numerically.

Formally, suppose $\bx \in \C^n$ is a discrete signal and that we are
given information about the squared modulus of the inner product
between the signal and some vectors $\bz_i$, namely, 
\begin{equation}
\label{eq:data}
b_i = |\<\bx,\bz_i\>|^2, \quad i = 1, \ldots, m.
\end{equation}
In truth, we would like to know $\<\bx,\bz_i\>$ and record both phase and
magnitude information but can only record the magnitude; in other
words, phase information is lost.  In the classical example discussed
above, the $\bz_i$'s are complex exponentials at frequency $\omega_i$ so
that one collects the squared modulus of the Fourier transform of $\bx$.
Of course, many other choices for the measurement vectors $\bz_i$ are
frequently discussed in the literature, see \cite{Fin04,BBC09} for
instance.


We wish to recover $\bx$ from the data vector $\bb$, and suppose first
that $\bx$ is known to be real valued a priori. Then assuming that
$\bx$ is uniquely determined by $\bb$ up to a global sign, the
recovery may be cast as a combinatorial optimization problem: find a
set of signs $\sigma_i$ such that the solution to the linear equations
$ \left<\bx,\bz_i\right> = \sigma_i \sqrt{b_i}$, call it $\hat{\bx}$,
obeys $|\left<\hat{\bx},\bz_i\right>|^2 = b_i$.  Clearly, there are
$2^m$ choices for $\sigma_i$ and only two choices of these signs yield
$\bx$ up to global phase.  The complex case is harder yet, since
resolving the phase ambiguities now consists of finding a collection
$\sigma_i$ of complex numbers, each being on the unit
circle. Formalizing matters, it has been shown that at least one
version of the phase retrieval problem is NP-hard~\cite{SC91}.  Thus,
one of the major challenges in the field is to find conditions on $m$
and $\bz_i$ which guarantee efficient numerical recovery.


A frame-theoretic approach to signal recovery from magnitude
measurements has been proposed in \cite{BCE06,BBC07,BBC09}, where the
authors derive various necessary and sufficient conditions for the
uniqueness of the solution, as well as various polynomial-time
numerical algorithms for very specific choices of $\bz_i$.  While
theoretically quite appealing, the drawbacks are that the methods are
(1) either algebraic in nature, thus severely limiting their stability
in the presence of noise or slightly inexact data, or (2) the number
$m$ of measurements is on the order of $n^2$, which is much too large
compared to the number of unknowns.

This paper follows a very different route and establishes that if the
vectors $\bz_i$ are independently and uniformly sampled on the unit
sphere, then our signal can be recovered exactly from the magnitude
measurements \eqref{eq:data} by solving a simple convex program we
introduce below; this holds with high probability with the proviso
that the number of measurements is on the order of $n \log n$. Since
there are $n$ complex unknowns, we see that the number of samples is
nearly minimal.  To the best of our knowledge, this is the first
result establishing that under appropriate conditions, the
computationally challenging nonconvex problem of reconstructing a
signal from magnitude measurements is formally equivalent to a convex
program in the sense that they are guaranteed to have the same unique
solution.  

Finally, our methodology is robust with respect to noise in the
measurements. To be sure, when the data are corrupted by a small
amount of noise, we also prove that the recovery error is small.

\subsection{Methodology} 
\label{sec:method}

We introduce some notation that shall be used throughout to explain
our methodology. Letting $\cA$ be the linear transformation
\begin{equation}
\label{linmap}
\begin{array}{lll}
  \mathcal{H}^{n \times n} & \goto & \R^m\\
  \bX & \mapsto & \{\bz_i^* \bX \bz_i\}_{1 \le i \le m}
\end{array}
\end{equation}
which maps Hermitian matrices into real-valued vectors, one can
express the data collection $b_i = |\<\bx, \bz_i\>|^2$ as
\begin{equation}
\label{matrixmeasurements}
\bb = \mathcal{A}(\bx \bx^*). 
\end{equation}
For reference, the adjoint operator $\cA^*$ maps real-valued inputs
into Hermitian matrices, and is given by
\[
\begin{array}{lll}
  \R^{m} & \goto & \mathcal{H}^{n \times n}\\
  \by & \mapsto & \sum_i y_i \, \bz_i \bz_i^*. 
\end{array}
\]

As observed in \cite{CMP10,CESV} (see also~\cite{LV11}),
the phase retrieval problem can be cast as the matrix recovery problem
\begin{equation}
\label{eq:rankmin}
  \begin{array}{ll}
    \text{minimize}   & \quad \rank(\bX)\\
    \text{subject to} & \quad  \cA(\bX) = \bb\\
& \quad \bX \succeq 0.
\end{array}
\end{equation}
Indeed, we know that a rank-one solution exists so the optimal $\bX$ has
rank at most one. We then factorize the solution as $\bx \bx^*$ in order
to obtain solutions to the phase-retrieval problem. This gives $\bx$ up
to multiplication by a unit-normed scalar. This is all we can hope for
since if $\bx$ is a solution to the phase retrieval problem, then $c \bx$
for any scalar $c \in \C$ obeying $|c| = 1$ is also
solution.\footnote{When the solution is unique up to multiplication by
  such a scalar, we shall say that unicity holds up to global phase.}

Rank minimization is in general NP hard, and we propose, instead,
solving a trace-norm relaxation. Although this is a fairly standard
relaxation in control \cite{Beck98,Mesbahi97}, the idea of casting the
phase retrieval problem as a trace-minimization problem over an
affine slice of the positive semidefinite cone is very recent
\cite{CMP10,CESV}. Formally, we suggest solving
\begin{equation}
\label{eq:tracemin}
 \begin{array}{ll}
    \text{minimize}   & \quad \tr(\bX)\\ 
    \text{subject to} & \quad  \cA(\bX) = \bb\\
& \quad \bX \succeq 0. 
\end{array}
\end{equation}
If the solution has rank one, we factorize it as above to recover our
signal. This method which lifts up the problem of vector recovery from
quadratic constraints into that of recovering a rank-one matrix from
affine constraints via semidefinite programming is known under the
name of {\em PhaseLift} \cite{CESV}.


The program \eqref{eq:tracemin} is a semidefinite program (SDP) in
standard form, and there is a rapidly growing list of algorithms for
solving problems of this kind as efficiently as possible.  The crucial
question is whether and under which conditions the combinatorially
hard problem \eqref{eq:rankmin} and the convex problem
\eqref{eq:tracemin} are formally equivalent.

\subsection{Main result}
\label{sec:mainresult}

In this paper, we consider the simplest and perhaps most natural model
of measurement vectors. In this statistical model, we simply assume
that the vectors $\bz_i$ are independently and uniformly distributed on
the unit sphere of $\C^n$ or $\R^n$. To be concrete, we distinguish
two models. 
\begin{itemize}
\item {\em The real-valued model.} Here, the unknown signal $\bx$ is
  real valued and the $\bz_i$'s are independently sampled on the unit
  sphere of $\R^n$. 
\item {\em The complex-valued model.} The signal $\bx$ is now complex
  valued and the $\bz_i$'s are independently sampled on the unit sphere
  of $\C^n$.
\end{itemize}

Our main result is that the convex program recovers $\bx$ exactly (up to
global phase) provided the number $m$ of magnitude measurements is on
the order of $n \log n$. 

\begin{theorem}
\label{teo:main}
Consider an arbitrary signal $\bx$ in $\R^n$ or $\C^n$ and suppose
that the number of measurements obeys $m \ge c_0 \, n \log n$, where
$c_0$ is a sufficiently large constant. Then in both the real and
complex cases, the solution to the trace-minimization program is exact
with high probability in the sense that \eqref{eq:tracemin} has a
unique solution obeying
\begin{equation}
  \label{eq:main}
  \hat \bX = \bx \bx^*.  
\end{equation}
This holds with probability at least $1 - 3e^{-\gamma \frac{m}{n}}$,
where $\gamma$ is a positive absolute constant. 
\end{theorem}

Expressed differently, Theorem~\ref{teo:main} establishes a rigorous
equivalence between a class of phase retrieval problems and a class of
semidefinite programs.  Clearly, any phase retrieval algorithm, no
matter how complicated or intractable, would need at least $2n$
quadratic measurements to recover a complex valued object $\bx \in
\C^n$.  In fact recent results, compare Theorem II in \cite{Fin04},
show that for complex-valued signals, one needs at least $3n-2$
intensity measurements to guarantee uniqueness of the solution to
\eqref{eq:rankmin}.  Further, Balan, Casazza and Edidin have shown
that with probability 1, $4n-2$ generic measurement vectors (which
includes the case of random uniform vectors) suffice for uniqueness
in the complex case~\cite{BCE06}. Hence, Theorem \ref{teo:main} shows that the
oversampling factor for perfect recovery via convex optimization is
rather minimal.

To be absolutely complete, we would like to emphasize that our
discrete signals $\bx$ may represent 1D, 2D, 3D and higher dimensional
objects. For instance, in 2D the vector $\bx \in \C^n$ might be a
family of samples of the form $x[t_1,t_2]$, $1 \le t_1 \le n_1$, $1
\le t_2 \le n_2$, and with $n = n_1 n_2$, so that $\bx$ is a discrete
2D image. In this case, we would record the squared magnitudes of the
dot product
\[
\<\bx,\bz_i\> = \sum_{t_1, t_2} \bar{x}[t_1,t_2] z_i[t_1,t_2]. 
\]
Hence, our framework and theory apply to one- or multi-dimensional
signals.

\subsection{Geometry}

We find it rather remarkable that the only solution to
\eqref{eq:tracemin} is $\hat \bX = \bx\bx^*$. To see why this is perhaps
unexpected, suppose for simplicity that the trace of the solution were
known (we might be given some side information or just have additional
measurements giving us this information) and equal to $1$, say. In
this case, the objective functional is of course constant over the
feasible set, and our problem reduces to solving the feasibility
problem
\begin{equation}
  \label{eq:feasibility}
  \begin{array}{ll} \text{find} & \quad \bX \\
    \text{such that} & \quad \cA(\bX) = \bb, \,\,  \bX \succeq 0
\end{array}
\end{equation}
with again the proviso that knowledge of $\cA(\bX)$ determines
$\tr(\bX)$ (equal to $\tr(\bx\bx^*) = \|\bx\|_2 = 1$).  In this
context, our main theorem states that $\bx\bx^*$ is the unique
feasible point. In other words, there is no other positive
semidefinite matrix $\bX$ in the affine space $\cA(\bX) =
\bb$. Naively, we would not expect this affine space of enormous
dimension---it is of co-dimension about $n \log n$ and thus of
dimension $n^2 - O(n \log n)$ in the complex case---to intersect the
positive semidefinite cone in only one point. Indeed, counting degrees
of freedom suggests that there are infinitely many candidates in the
intersection.  The reason why this is not the case, however, is
precisely because there is a feasible solution with low rank. Indeed,
the slice of the positive semidefinite cone $\{\bX : \bX \succeq 0\}
\cap \{\tr(\bX) = 1\}$ is quite `pointy' at $\bx\bx^*$ and it is,
therefore, possible for the affine space $\{\cA(\bX) = \bb\}$ to be
tangent even though it is of very small codimension.

Figure \ref{fig:psd} represents this geometry. In this example, 
\[
\bx = \frac{1}{\sqrt2} \begin{bmatrix} 1 \\ -1 \end{bmatrix} \quad
\Longrightarrow \quad \bx\bx^* = \frac12 \begin{bmatrix} 1 & -1 \\
  -1 & 1 \end{bmatrix}
\] 
and the affine space $\cA(\bX) = \bb$ is tangent to the positive
semidefinite cone at the point $\bx \bx^*$.
\begin{figure}
\centering
\includegraphics[width=60mm]{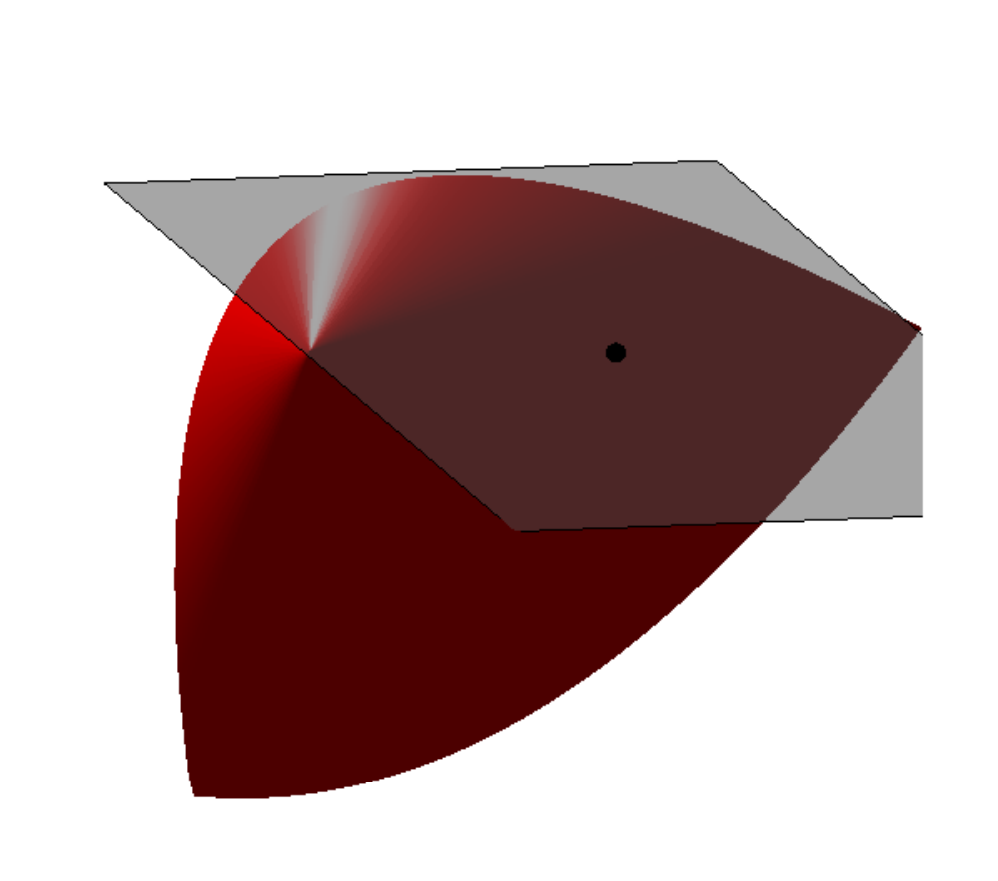}
\includegraphics[width=60mm]{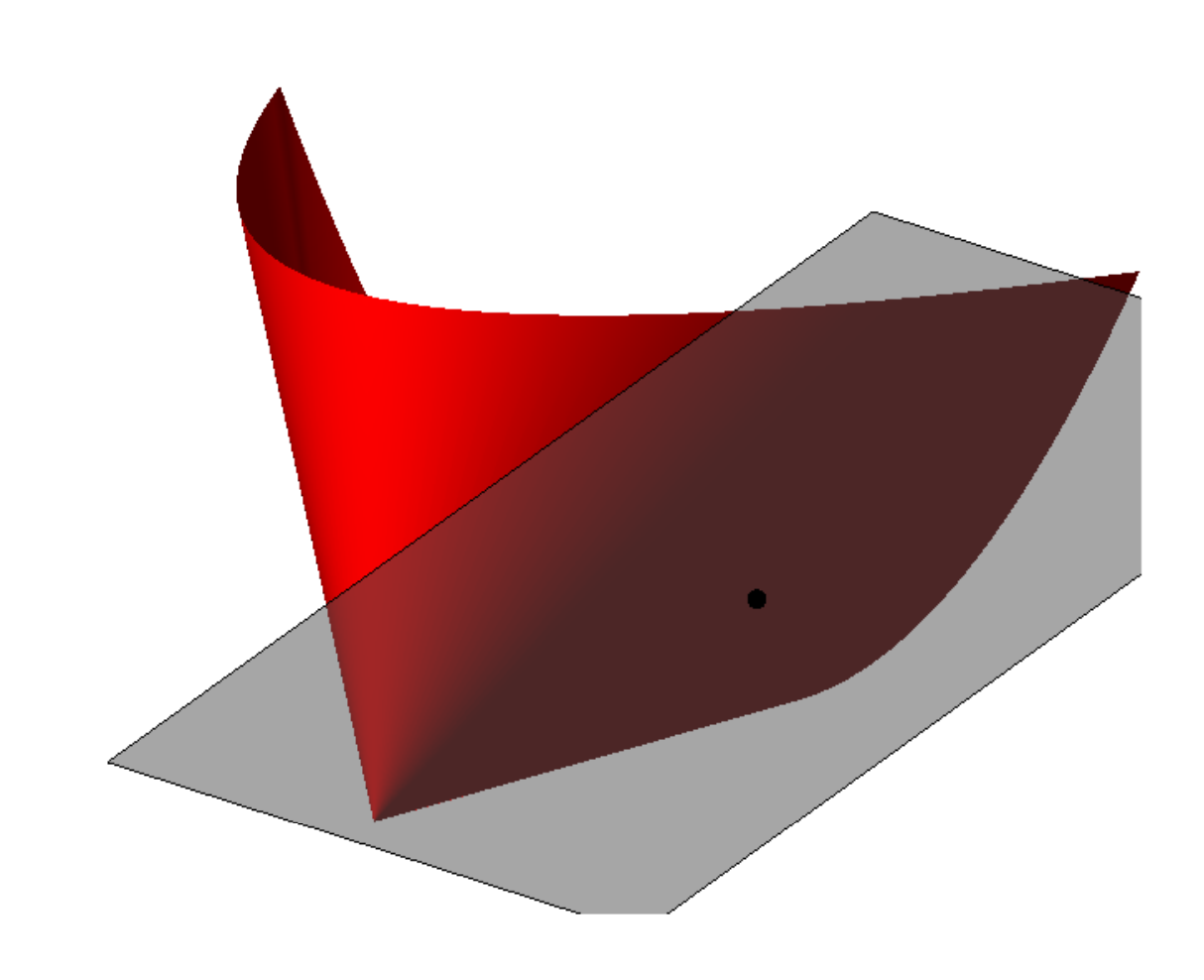}
\caption[]{Representation of the affine space $\cA(\bX) = \bb$ (gray) and of
  the semidefinite cone
$\begin{bmatrix} x & y \\ y & z \end{bmatrix}
  \succeq 0$ (red) which is a subset of $\R^3$. These two sets are
  drawn so that they are tangent to each other at the rank 1 matrix
  $\frac{1}{2} \begin{bmatrix} 1 & -1 \\
  -1 & 1 \end{bmatrix}$ (black dot).  Two views of the same 3D
 figure are provided for convenience.
 }
\label{fig:psd}
\end{figure}

Mathematically speaking, phase retrieval is a problem in algebraic
geometry since we are trying to find a solution to a set of polynomial
equations. The originality in our approach is that we do not use tools
from this field. For instance, we prove that there is no other
positive semidefinite matrix $\bX$ in the affine space $\cA(\bX) =
\bb$, or equivalently, that a certain system of polynomial equations
(a symmetric matrix is positive semidefinite if and only if the determinants of
all the leading principal minors are nonnegative) only has one
solution; this is a fact that general techniques from algebraic
geometry appear to not detect.

\subsection{Stability}
\label{subsec:stability}

In the real world, measurements are contaminated by noise. Using the
frameworks developed in~\cite{MCNoise} and \cite{QSTC}, it is possible to extend
Theorem~\ref{teo:main} to accommodate noisy measurements. One could
consider a variety of noise models as discussed in \cite{CESV} but we
work here with a simple generic model in which we observe 
\begin{equation}
\label{eq:noisydata}
b_i = |\<\bx,\bz_i\>|^2 + \nu_i, 
\end{equation}
where $\nu_i$ is a noise term with bounded $\ell_2$ norm,
$\|\vct{\nu}\|_2 \le \epsilon$.  This model is nonstandard since the
usual statistical linear model posits a relationship of the form $b_i
= \<\bx,\bz_i\> + \nu_i$ in which the mean response is a linear
function of the unknown signal, not a quadratic function.
Furthermore, we prefer studying \eqref{eq:noisydata} rather than the
related model $b_i = |\<\bx,\bz_i\>| + \nu_i$ (the modulus is not
squared) because in many applications of interest in optics and other
areas of physics, one can measure squared magnitudes or
intensities---not magnitudes.

We now consider the solution to 
\begin{equation}
\label{eq:traceminnoisy}
 \begin{array}{ll}
   \text{minimize}   & \quad \tr(\bX)\\ 
   \text{subject to} & \quad  \|\cA(\bX) - \bb\|_2 \le \epsilon\\
   & \quad \bX \succeq 0. 
\end{array}
\end{equation}
We do not claim that $\hat \bX$ has low rank so we suggest estimating
$\bx$ by extracting the largest rank-1 component. Write $\hat \bX$ as
\newcommand{\hlambda}{\hat{\lambda}}
\[
\hat \bX = \sum_{k = 1}^n \hlambda_k \hat{\bu}_k \hat{\bu}_k^*, \quad
\hlambda_1  \ge \ldots \ge \hlambda_n \ge 0,  
\]
and set 
\[
\hat{\bx} = \sqrt{\hlambda_1} \, \hat{\bu}_1.
\]
We prove the following estimate.
\begin{theorem}
\label{teo:stability}
Fix $\bx \in \C^n$ or $\R^n$ and assume the $\bz_i$'s are uniformly
sampled on the sphere of radius $\sqrt{n}$. Under the hypotheses of
Theorem \ref{teo:main}, the solution to \eqref{eq:traceminnoisy} obeys
($\|\bX\|_2$ is the Frobenius norm of $\bX$)
\begin{equation}
  \label{eq:noisy}
  \|\hat \bX - \bx\bx^*\|_2 \le C_0 \, \epsilon 
\end{equation}
for some positive numerical constant $C_0$. We also have 
\begin{equation}
  \label{eq:noisy2}
  \|\hat \bx - e^{i\phi} \bx\|_2 \le C_0 \,  \text{\em min}(\|\bx\|_2,
\epsilon/\|\bx\|_2)
\end{equation}
for some $\phi \in [0,2\pi]$. Both these estimates hold with nearly
the same probability as in the noiseless case. 
\end{theorem}




Thus our approach also provides stable recovery in presence of
noise. This important property is not shared by other reconstruction
methods, which are of a more algebraic nature and rely on particular
properties of the measurement vectors, such as the methods
in~\cite{Fin04,BCE06,BBC09}, as well as the methods that appear
implicitly in Theorem~3.1 and Theorem~3.3 of~\cite{CESV}.

We note that one can further improve the accuracy of the solution
$\hat{\bx}$ by ``debiasing'' it. We replace $\hat{\bx}$ by its
rescaled version $s \hat{\bx}$ where $s = \sqrt{\sum_{k=1}^{n}
  \hat{\lambda}_k}/\|\hat{\bx}\|_2$.  This corrects for the energy
leakage occurring when $\hat{\bX}$ is not exactly a rank-1 solution,
which could cause the norm of $\hat{\bx}$ to be smaller than that of
the actual solution. Other corrections are of course possible. 

\subsection{Organization of the paper}

The remainder of the paper is organized as follows.
Subsection~\ref{subsec:notation} introduces some notation used
throughout the paper. In Section~\ref{sec:proof} we present the main
architecture of the proof of Theorem~\ref{teo:main}, which comprises
two key ingredients: approximate $\ell_1$ isometries and approximate
dual certificates.  Section~\ref{sec:isometries} is devoted to
establishing approximate $\ell_1$ isometries. In
Section~\ref{sec:dual}, we construct approximate dual certificates and
complete the proof of Theorem~\ref{teo:main} in the real-valued
case. Section~\ref{sec:complex} shows how the proof for the
real-valued case can be adapted to the complex-valued
case. Section~\ref{sec:stability} is concerned with the proof of
Theorem~\ref{teo:stability}. Numerical simulations, illustrating our
theoretical results, are presented in Section~\ref{sec:numerics}. We
conclude the paper with a short discussion in Section
\ref{sec:discussion}.

\subsection{Notations}
\label{subsec:notation}

\newcommand{\Tp}{T^\perp}
\newcommand{\PT}{\mathcal{P}_T}
\newcommand{\PTp}{\mathcal{P}_{T^\perp}}

It is useful to introduce notations that shall be used throughout the
paper. Matrices and vectors are denoted in boldface (such as $\bX$ or
$\bx$), while individual entries of a vector or matrix are denoted in
normal font; e.g.\ the $i$th entry of $\bx$ is $x_i$.  For matrices,
we define
\[
\|\bX\|_p = \Bigl[\sum_i \sigma_i^p(\bX)\Bigr]^{1/p},
\]
(where $\sigma_i(\bX)$ denotes the $i$th singular value of $\bX$), so that
$\|\bX\|_1$ is the nuclear norm, $\|\bX\|_2$ is the Frobenius norm
and $\|\bX\|_\infty$ is the operator norm also denoted by $\|\bX\|$. For
vectors, $\|\bx\|_p$ is the usual $\ell_p$ norm. We denote the $n-1$
dimensional sphere by $S^{n-1}$, i.e.~the set $\{\bx \in \R^n : \|\bx\|_2
= 1\}$.

Next, we define $T_{\bx}$ to be the set of symmetric matrices of the
form
\begin{equation}
\label{eq:T}
T_{\bx} = \{\bX = \bx \by^* + \by \bx^*: \by \in \R^n\} 
\end{equation}
and denote $T_{\bx}^\perp$ by its orthogonal complement.  Note that
$\bX \in T_{\bx}^\perp$ if and only if both the column and row spaces
of $\bX$ are perpendicular to $\bx$. Further, the operator
$\mathcal{P}_{T_{\bx}}$ is the orthogonal projector onto $T_{\bx}$ and
similarly for $\mathcal{P}_{T_{\bx}^\perp}$. We shall almost always
use $\bX_{T_{\bx}}$ as a shorthand for $\mathcal{P}_{T_{\bx}}(\bX)$.

Finally, we will abuse language and say that a symmetric matrix $\bH$
is feasible if and only if $\bx\bx^* + \bH$ is feasible for our
problem \eqref{eq:tracemin}. This means that $\bH$ obeys
\begin{equation}
\label{eq:feasible}
\bx\bx^* + \bH \succeq 0 \quad \text{and} \quad \cA(\bH) = 0. 
\end{equation}

\section{Architecture of the Proof}
\label{sec:proof}

In this section, we introduce the main architecture of the argument
and defer the proofs of crucial intermediate results to later
sections. 
We shall prove Theorem \ref{teo:main} in the real case
first for ease of exposition. Then in Section \ref{sec:complex}, we
shall explain how to modify the argument to the complex and more
general case.

Suppose then that $\bx \in \R^n$ and that the $\bz_i$'s are sampled on
the unit sphere.  It is clear that we may assume without loss of
generality that $\bx$ is unit-normed. Further, since the uniform
distribution on the unit sphere is rotationally invariant, it suffices
to prove the theorem in the case where $\bx = \be_1$. Indeed, we can
write any unit vector $\bx$ as $\bx = \bU \be_1$ where $\bU$ is
orthogonal. Since
\[
|\<\bx, \bz_i\>|^2 = |\<\bU\be_1, \bz_i\>|^2 = |\<\be_1, \bU^* \bz_i\>|^2 =^d
|\<\be_1, \bz_i\>|^2,
\]
the problem is the same as that of finding $\be_1$. We henceforth
assume that $\bx = \be_1$.

Finally, the theorem can be equivalently stated in the case where the
$\bz_i$'s are i.i.d.~copies of a white noise vector $\bz \sim \mathcal{N}(0,I)$
with independent standard normals as components. Indeed, if $\bz_i \sim
\mathcal{N}(0,I)$, 
\[
|\<\bx, \bz_i\>|^2 = b_i \quad \Longleftrightarrow \quad |\<\bx, \bu_i\>|^2 =
b_i/\|\bz_i\|_2^2,
\]
where $\bu_i = \bz_i/\|\bz_i\|_2$ is uniformly sampled on the unit
sphere. Since $\|\bz_i\|_2$ does not vanish with probability one,
establishing the theorem for Gaussian vectors establishes it for
uniformly sampled vectors and vice versa. From now on, we assume $\bz_i$
i.i.d.\ $\mathcal{N}(0,I)$.

\subsection{Key lemma}

The set $T := T_{\be_1}$ defined in \eqref{eq:T} may be interpreted as
the tangent space at $\be_1 \be_1^*$ to the manifold of symmetric
matrices of rank 1. Now standard duality arguments in semidefinite
programming show that a sufficient (and nearly necessary) condition
for $\bx\bx^*$ to be the unique solution to \eqref{eq:tracemin} is
this:
\begin{itemize}
\item the restriction of $\cA$ to $T$ is injective ($\bX \in T$ and
  $\cA(X) = 0 \Rightarrow \bX = 0$),
\item and there exists a {\em dual certificate} $\bY$ in the range of
  $\cA^*$ obeying\footnote{The notation $A \prec B$ means that $B - A$
    is positive definite.}
\begin{equation}
  \label{eq:dual}
  \bY_T = \be_1 \be_1^* \quad \text{and} \quad \bY_\Tp \prec I_\Tp. 
\end{equation}
\end{itemize}
The proof is straightforward and omitted.  Our strategy to prove
Theorem \ref{teo:main} hinges on the fact that a strengthening of the
injectivity property allows to relax the properties of the dual
certificate, as in the approach pioneered in \cite{Gross09} for matrix
completion. We establish the crucial lemma below.
\begin{lemma}
  \label{lem:crucial}
  Suppose that the mapping $\cA$ obeys the following two properties:
  for all positive semidefinite matrices $X$,
\begin{equation}
\label{eq:RIP1r1}
m^{-1} \|\cA(\bX)\|_1 < (1+1/9) \|\bX\|_1;
\end{equation}
and for all matrices $\bX \in T$
\begin{equation}
\label{eq:RIP1r2}
m^{-1} \|\cA(\bX)\|_1 > 0.94(1-1/9) \|\bX\|. 
\end{equation}
Suppose further that there exists $\bY$ in the range of $\cA^*$ obeying
\begin{equation}
  \label{eq:dualcertif}
  \|\bY_T - \be_1 \be_1^*\|_2 \le 1/3 \quad \text{and} \quad \|\bY_\Tp\| \le 1/2.
\end{equation}
Then $\be_1 \be_1^*$ is the unique minimizer to \eqref{eq:tracemin}.
\end{lemma}

The first property \eqref{eq:RIP1r1} is reminiscent of the (one-sided)
RIP property in the area of compressed sensing \cite{CT05}. The
difference is that it is expressed in the 1-norm rather than the
2-norm. Having said this, we note that RIP-1 properties have also been
used in the compressed sensing literature, see \cite{IndykGilbert} for
example. We use this property instead of a property about
$\|\cA(\bX)\|_2$, because we actually believe that a RIP property in
the 2-norm does not hold here because $\|\cA(\bX)\|_2^2$ involves fourth
moments of Gaussian variables. The second property \eqref{eq:RIP1r2}
is a form of local RIP-1 since it holds only for matrices in $T$.

We would like to emphasize that the bound for the dual certificate in
\eqref{eq:dualcertif} is loose in the sense that $\bY_T$ and $\be_1 \be_1^*$
may not be that close, a fact which will play a crucial role in our
proof.  This is in stark contrast with the work of David
Gross~\cite{Gross09}, which requires a very tight approximation.

\subsection{Proof of Lemma \ref{lem:crucial}}

We need to show that there is no feasible $\bx\bx^* + \bH \neq \bx\bx^*$ with
$\tr(\bx\bx^* + \bH) \leq \tr(\bx\bx^*)$.  Consider then a feasible $\bH \neq 0$
obeying $\tr(\bH) \le 0$, write
\[
\bH = \bH_T + \bH_\Tp,
\]
and observe that 
\begin{equation}
\label{eq:AH}
0 = \|\cA(\bH)\|_1 = \|\cA(\bH_T)\|_1 - \|\cA(\bH_\Tp)\|_1.  
\end{equation}
Now it is clear that $\bx\bx^* + \bH \succeq 0 \Rightarrow \bH_\Tp \succeq 0$
and, therefore, \eqref{eq:RIP1r1} gives
\[
m^{-1} \|\cA(\bH_\Tp)\|_1 \le (1+\delta) \tr(\bH_\Tp)
\]
for some $\delta < 1/9$.  Also, $\tr(\bH_T) \le -\tr(\bH_\Tp) \le 0$, which
implies that $\left|\tr(\bH_T)\right| \geq \tr(\bH_\Tp)$.  We then show
that the operator and Frobenius norms of $\bH_T$ must nearly be the
same.
\begin{lemma}
  \label{lem:RIP}
  Any feasible matrix $\bH$ such that $\tr(\bH) \le 0$ must obey
\[
\|\bH_T\|_2 \le \sqrt{\frac{17}{16}} \, \|\bH_T\|.
\]
\end{lemma}
\begin{proof}
  Since the matrix $\bH_T$ has rank at most 2 and cannot be negative
  definite, it is of the form
\[
-\lambda (\bu_1 \bu_1^* - t \bu_2 \bu_2^*),  
\]
where $\bu_1$ and $\bu_2$ are orthonormal eigenvectors, $\lambda \ge 0$
and $t \in [0,1]$. We claim that we cannot have $t \ge
1/4$.\footnote{The choice of $1/4$ is somewhat arbitrary here.}
Suppose the contrary and fix $t \ge 1/4$. By \eqref{eq:RIP1r2}, we
know that
\[
m^{-1} \, \|\cA(\bH_T)\|_1 \ge 0.94 (1-\delta) \|\bH_T\|.
\]
Further, since
\[
\|\bH_T\| = \frac{|\tr(\bH_T)|}{1-t} \ge \frac{4}{3} |\tr(\bH_T)|
\]
for $t \ge 1/4$, it holds that 
\[
0 \geq \frac{5}{4}(1-\delta)\left|\tr(\bH_T) \right| -
(1+\delta)\tr(\bH_\Tp). 
\]
The right-hand side above is positive if $\tr(\bH_\Tp) <
\frac{5}{4}\frac{(1-\delta)}{(1+\delta)}\left|\tr(\bH_T)\right|$, so
that we may assume that
\[
\tr(\bH_\Tp) \geq \frac{5}{4}\frac{(1-\delta)}{(1+\delta)}\left|\tr(\bH_T)\right|. 
\]
Since, $\left|\tr(\bH_T)\right| \geq \tr(\bH_\Tp)$, this gives
\[
0 \ge \Bigl[\frac{5}{4}(1-\delta) - (1+\delta)\Bigr] \tr(\bH_\Tp).
\]
If $\delta < 1/9$, the only way this can happen is if $\tr(\bH_\Tp) = 0
\Rightarrow \bH_\Tp = 0$. So we would have $\bH = \bH_T$ of rank 2 and
$\cA(\bH_T) = 0$. Clearly, \eqref{eq:RIP1r2} implies that $\bH = 0$.

Now that it is established that $t \le 1/4$, the chain of inequalities
follow from the relation between the eigenvalues of $\bH_T$. 
\end{proof}

To conclude the proof of Lemma \ref{lem:crucial}, we show that the
existence of an inexact dual certificate rules out the existence of
matrices obeying the conditions of Lemma \ref{lem:RIP}. From
\[
0.94(1-\delta) \|\bH_T\| \le \|\cA(\bH_T)\|_1 = \|\cA(\bH_\Tp)\|_1 \le
(1+\delta) \tr(\bH_\Tp),
\]
we conclude that 
\begin{equation}
\label{eq:new}
\tr(\bH_\Tp) \ge 0.94 \frac{1-\delta}{1+\delta} \, \|\bH_T\| \ge 0.94
\frac{1-\delta}{1+\delta} \sqrt{\frac{16}{17}} \|\bH_T\|_2,
\end{equation}
where we used Lemma \ref{lem:RIP}. Next, 
\begin{align*}
0 \ge \tr(\bH_T) + \tr(\bH_\Tp) & = \<\bH, \be_1 \be_1^*\> + \tr(\bH_\Tp) \\
& = \<\bH, \be_1 \be_1^* - \bY\> + \<\bH,\bY\> +   \tr(\bH_\Tp)\\
& =  \<\bH_T, \be_1 \be_1^* - \bY_T\> - \<\bH_\Tp, \bY_\Tp\> +  \tr(\bH_\Tp)\\
& \ge \frac12 \tr(\bH_\Tp) - \frac{1}{3} \|\bH_T\|_2.  
\end{align*}
The third line above follows from $\<\bH, \bY\> = 0$ and the fourth from
Cauchy-Schwarz together with $|\<\bH_\Tp, \bY_\Tp\>| \le \frac12
\tr(\bH_\Tp)$. Hence, it follows from \eqref{eq:new} that 
\[
0 \ge \frac12 \Bigl(0.94 \frac{1-\delta}{1+\delta}
\sqrt{\frac{16}{17}} - \frac{2}{3}\Bigr) \|\bH_T\|_2.
\]
Since the numerical factor is positive for $\delta < 0.155$, the only
way this can happen is if $\bH_T = 0$. In turn,
$\|\mathcal{A}(\bH_\Tp)\|_1 = 0 \ge (1-\delta) \tr(\bH_\Tp)$ which gives
$\bH_\Tp = 0$. This concludes the proof.

\section{Approximate $\ell_1$ Isometries}
\label{sec:isometries}

We have seen that in order to prove our main result, it suffices to
show 1) that the measurement operator $\cA$ enjoys approximate isometry
properties (in an $\ell_1$ sense) when acting on low-rank matrices and
2) that an inexact dual certificate exists. This section focuses on
the former and establishes that both \eqref{eq:RIP1r1} and
\eqref{eq:RIP1r2} hold with high probability. In fact, we shall prove
stronger results than what is strictly required.

\begin{lemma}
  \label{lem:A1}
  Fix any $\delta > 0$ and assume $m \ge 16 \delta^{-2} \, n$. Then
  for all unit vectors $\bu$,
\begin{equation}
\label{eq:A1}
(1-\delta) \le \frac{1}{m} \|\cA(\bu\bu^*)\|_1 \le (1+\delta) 
\end{equation}
on an event $E_\delta$ of probability at least $1-2e^{-m
  \epsilon^2/2}$, where $\delta/4 = \epsilon^2 + \epsilon$. On the
same event,
\[
(1-\delta) \|\bX\|_1 \le \frac{1}{m} \|\cA(\bX)\|_1 \le (1+\delta) \|\bX\|_1 
\]
for all positive semidefinite matrices. The right inequality holds for
all matrices.
\end{lemma}
\begin{proof} This lemma has an easy proof. Let $\bZ$ be the $m \times
  n$ matrix with $\bz_i$'s as rows. Then 
\[
\|\cA(\bu\bu^*)\|_1 = \sum_i |\<\bz_i, \bu\>|^2 = \|\bZ \bu\|^2 
\]
so that 
\[
\sigma^2_{\text{min}}(\bZ) \le \|\cA(\bu\bu^*)\|_1 \le
\sigma^2_{\text{max}}(\bZ). 
\]
The claim is a consequence of well-known deviations bounds concerning
the singular values of Gaussian random matrices \cite{VershyninRMT},
namely,
\begin{align*}
\P\left(\sigma_{\text{max}} (\bZ) > \sqrt{m} + \sqrt{n} + t\right)  & \le e^{- t^2/2}\\
\P\left(\sigma_{\text{min}}(\bZ) < \sqrt{m} - \sqrt{n} - t\right)  & \le e^{- t^2/2}.
\end{align*}
The conclusion follows from taking $m \ge \epsilon^{-2} \, n$ and $t =
\sqrt{m} \epsilon$. For the second part of the lemma, observe that $\bX
= \sum_j \lambda_j \bu_j \bu_j^*$ with nonnegative eigenvalues $\lambda_j$
so that
\[
\|\cA(\bX)\|_1 = \sum_j \sum_i \lambda_j |\<\bu_j,\bz_i\>|^2 = \sum_j
\lambda_j \|\cA(\bu_j \bu_j^*)\|_1. 
\]
The claim follows from \eqref{eq:A1}. The last claim is a consequence
of $\|\cA(\bX)\|_1 \le \sum_j \sum_i |\lambda_j| |\<\bu_j,\bz_i\>|^2$
together with $\sum_j |\lambda_j| = \|\bX\|_1$.
\end{proof}


Our next result is concerned with the mapping of rank-2 matrices.
\begin{lemma}
  \label{lem:A2}
  Fix $\delta > 0$. Then there are positive numerical constants $c_0$
  and $\gamma_0$ such that if $m \ge c_0 \, [\delta^{-2} \log
  \delta^{-1}] \, n$, $\cA$ obeys the following property with
  probability at least $1 - 3 e^{-\gamma_0 m\delta^2}$:
  for any symmetric rank-2 matrix $\bX$,
\begin{equation}
\label{eq:A2b}
\frac{1}{m} \|\cA(\bX)\|_1 \ge 0.94 (1-\delta) \|\bX\|.  
\end{equation}
\end{lemma}
\begin{proof} 
  By homogeneity, it suffices to consider the case where $\|\bX\| =
  1$. Consider then a rank-2 matrix $\bX$ with eigenvalue decomposition
  $X = \bu_{1}\bu_{1}^*-t\bu_{2}\bu_{2}^*$ with $t\in \left[ -1, 1 \right]$
  and orthonormal $\bu_i$'s. Note that for $t \le 0$, Lemma \ref{lem:A1}
  already claims a tighter lower bound so it only suffices to consider
  $t \in \left[0, 1 \right]$. We have
\[
\frac1m \|\mathcal{A}(\bX)\|_{1} = \frac1m \sum_{i=1}^m \Bigl\vert
|\<\bu_1,\bz_i\>|^2 - t |\<\bu_2,\bz_i\>|^2 \Bigr\vert = \frac1m \sum_i \xi_i, 
\]
where the $\xi_i$'s are independent copies of the random variable 
\[
\xi = |Z_1^2 - t Z_2^2|
\]
in which $Z_1$ and $Z_2$ are independent standard normal
variables. This comes from the fact that $\<\bu_1,\bz_i\>$ and $\<\bu_2,
\bz_i\>$ are independent standard normal. We calculate below that
\begin{equation}
\label{eq:f}
\E \xi = f(t) = \frac{2}{\pi} \Bigl(2\sqrt{t} + (1-t)(\pi/2 - 2 \arctan(\sqrt{t}))\Bigr). 
\end{equation}
The graph of this function is shown in Figure \ref{fig:f}; we check
that $f(t) \ge 0.94$ for all $t \in [0,1]$.
\begin{figure}
\centering
    \includegraphics[width=0.49\textwidth]{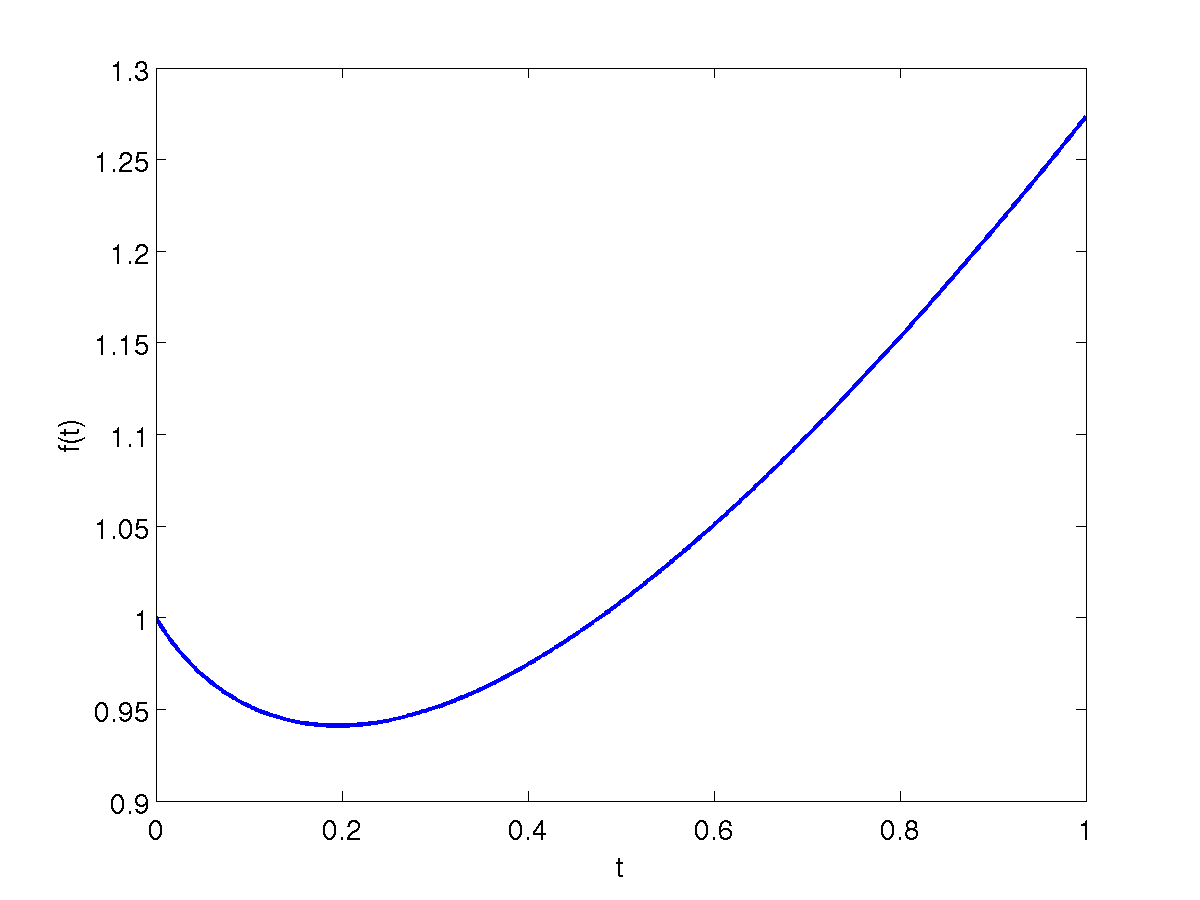}
    \caption{$f(t) = \E|Z_1^2 - t Z_2^2|$ as a function of $t$.}
\label{fig:f}
\end{figure}

We now need a deviation bound concerning the fluctuation of $m^{-1}
\sum_i \xi_i$ around its mean and this is achieved by classical
Chernoff bounds.  Note that $\xi \le Z_1^2 + |t| Z_2^2$ is a
sub-exponential variable and thus, $\|\xi\|_{\psi_1} := \sup_{p\geq 1}
\, [\E |\xi|^p]^{1/p}$ is finite.\footnote{It would be possible to
  compute a bound on this quantity but we will not pursue this at the
  moment.}
\begin{lemma}[Bernstein-type inequality \cite{VershyninRMT}]
\label{lem:bernsteintype}
  Let $X_1,\ldots, X_m$ be i.i.d.~sub-exponential random
  variables. Then
\[
\P\Bigl(\Bigl\vert \frac{1}{m} \sum_{i = 1}^m X_i - \E X_1\Bigr\vert
\ge \epsilon\Bigr) \leq 2 \exp\Bigl[ - c_0 \, m
\min\Bigl(\frac{\epsilon^2}{\|X\|^2_{\psi_1}},
\frac{\epsilon}{\|X\|_{\psi_1}}\Bigr)\Bigr]
\]
in which $c_0$ is a positive numerical constant. 
\end{lemma}

We have thus established that for a fixed $X$, 
\[m^{-1} \|\cA(\bX)\|_1 \ge (0.94 - \epsilon_0) \|\bX\|
\]
with probability at least $1 - 2e^{-\gamma_0 m \epsilon_0^2}$
(provided $\epsilon_0 \le \|\xi\|_{\psi_1}$, which we assume). 

To complete the argument, let $\mathcal{S}_{\epsilon}$ be an
$\epsilon$ net of the unit sphere, $\mathcal{T}_{\epsilon}$ be an
$\epsilon$ net of $\left[0,1\right]$, and set $$ \mathcal{N}_\epsilon =
\{\bX = \bu_{1}\bu_{1}^* -t \bu_{2} \bu_{2}^* : (\bu_{1},\bu_{2},t) \in
\mathcal{S}_{\epsilon} \times \mathcal{S}_{\epsilon} \times
\mathcal{T}_{\epsilon} \}.$$ Since $|\mathcal{S}_{\epsilon}| \le
(3/\epsilon)^n$, we have
\[
|\mathcal{N}_\epsilon| \le (3/\epsilon)^{2n+1}. 
\]
Now for any $\bX = \bu \bu ^* - t \bv\bv^*$, consider the approximation $\bX_0 =
\bu_0 \bu_0^* - t_0 \bv_0 \bv_0^* \in \mathcal{N}_\epsilon$, where $\|\bu_0 -
\bu\|_2$, $\|\bv - \bv_0\|_2$ and $|t - t_0|$ are each at most
$\epsilon$. We claim that 
\begin{equation}
  \label{eq:9eps}
  \|\bX-\bX_0\|_1 \le 9 \epsilon,  
\end{equation}
and postpone the short proof. On the intersection of $E_1 = \{m^{-1}
\|\cA(\bX)\|_1 \le (1+\delta_1)\|X\|_1, \text{ for all $\bX$}\}$ with $E_2
:= \{m^{-1}\|\cA(\bX_0)\|_1 \geq (0.94-\epsilon)\|\bX_0\|, \text{ for all
} \bX_0 \in \mathcal{N}_\epsilon \}$,
\begin{align*}
  m^{-1} \|\cA(\bX)\|_1 & \ge \|\cA(\bX_0)\|_1 - \|\cA(\bX-\bX_0)\|_1\\
  &\ge (0.94-\epsilon)\|\bX_0\| - 9(1+\delta_1)\epsilon\\
  & \ge (0.94-\epsilon)(\|\bX\|-\|\bX_0-\bX\|) -9(1+\delta_1)\epsilon\\
  & \ge (0.94-\epsilon)(1-5\epsilon) -9(1+\delta_1)\epsilon\\
  & \ge 0.94 - (15+9\delta_1)\epsilon,
\end{align*}
which is the desired bound by setting $0.94 \delta =
(15+9\delta_1)\epsilon$. In conclusion, set $\delta_1 = 1/2$ and take
$\epsilon = 0.94 \delta/20$. Then $E_1$ holds with probability at
least $1 - O(e^{-\gamma_1 m \epsilon^2})$ provided $m$ obeys the
condition of the theorem. Further, Lemma \ref{lem:A2} states that
$E_2$ holds with probability at least $1 - 2e^{-\gamma_2 m}$. This
concludes the proof provided we check \eqref{eq:9eps}.

We begin with 
\[
\|\bX-\bX_0\|_1 \le \|\bu\bu^* - \bu_0 \bu_0^*\|_1 + |t-t_0| \|\bv\bv^*\|_1 + |t_0|\|
\bv\bv^* - \bv_0\bv_0^*\|_1.
\]
Now 
\[
\|\bu\bu^* - \bu_0\bu_0^*\|_1 \le 2 \|\bu\bu^* - \bu_0\bu_0^*\| \le 4\|\bu-\bu_0\|_2, 
\]
where the first inequality follows from the fact that $\bu\bu^* -
\bu_0\bu_0^*$ is of rank at most 2, and the second follows from 
\begin{align*}
  \|\bu\bu^* - \bu_0\bu_0^*\| & = \sup_{\|\bx\|_2 = 1} \, \Bigl\vert \<\bu_0,\bx\>^2 -
  \<\bu,\bx\>^2\Bigr\vert \\
  & = \sup_{\|\bx\|_2 = 1} \, \Bigl\vert \<\bu - \bu_0,\bx\>
  \<\bu+\bu_0,\bx\>\Bigr\vert \le \|\bu - \bu_0\|_2 \|\bu + \bu_0\|_2 \le 2 \|\bu -
  \bu_0\|_2.
\end{align*}
Similarly, $\|\bv\bv^* - \bv_0\bv_0^*\|_1 \le 4\epsilon$ and this concludes
the proof.\footnote{The careful reader will remark that we have also
  used $\|\bX-\bX_0\| \le 5\epsilon$, which also follows from our
  calculations.}
\end{proof}

\begin{lemma}
  Let $Z_1$ and $Z_2$ be independent $\mathcal{N}(0,1)$ variables and
  $t \in [0,1]$. We have
\[
E |Z_1^2 - t Z_2^2| = f(t), 
\] 
where $f(t)$ is given by \eqref{eq:f}.
\end{lemma}
\begin{proof}
  Set
\[
\rho = \frac{1-t}{1+t} \text{ and } \cos
\theta = \rho
\]
in which $\theta \in [0, \pi/2]$.  By using polar coordinates, we have
\begin{align*}
  \E |Z_1^2 - t Z_2^2| & = \frac{1}{2\pi} \int_0^\infty r^3 e^{-r^2/2} \, dr \, \int_0^{2\pi} | \cos^2 \phi - t \sin^2\phi| \, d\phi\\
& =   \frac{1}{\pi} \, \int_0^{2\pi} | \cos^2 \phi - t \sin^2\phi| \, d\phi\\
&  =   \frac{2}{\pi} \, \int_0^{\pi} | \cos^2 \phi - t \sin^2\phi| \, d\phi
\end{align*}
Now using the identities $\cos^2 \phi = (1+\cos 2\phi)/2$ and
$\sin^2\phi = (1-\cos 2\phi)/2$, we have
\begin{align*}
  \E |Z_1^2 - t Z_2^2| & = \frac{1+t}{\pi} \, \int_0^\pi |\cos 2\phi +
  \rho| \, d\phi\\
  & = \frac{1+t}{2\pi} \, \int_0^{2\pi} |\cos \phi +
  \rho| \, d\phi\\
  & = \frac{1+t}{\pi} \, \int_0^{\pi} |\cos \phi +
  \rho| \, d\phi\\
  & = \frac{1+t}{\pi} \, \int_0^{\pi} |\rho - \cos \phi|
  \, d\phi\\
  & = \frac{1+t}{\pi} \,\Bigl[ \int_0^\theta \cos \phi - \rho \, d\phi
  + \int_\theta^\pi \rho - \cos \phi \, d\phi \Bigr]\\
  & = \frac{2}{\pi} (1+t) [\sin \theta + \rho (\pi/2 - \theta)].
\end{align*}
We recognize $\eqref{eq:f}$. 
\end{proof}

\section{Dual Certificates} 
\label{sec:dual}

To prove our main theorem, it remains to show that one can construct
an inexact dual certificate $\bY$ obeying the conditions of Lemma
\ref{lem:crucial}.

\subsection{Preliminaries}

\newcommand{\cS}{\mathcal{S}}

The linear mapping $\cA^*\cA$ is of the form\footnote{For symmetric
  matrices, $\mtx{A} \otimes \mtx{B}$ is the linear mapping $\bH
  \mapsto \<\mtx{A}, \bH\> \mtx{B}$.}
\[
\cA^* \cA = \sum_{i =1}^m \bz_i \bz_i^* \otimes \bz_i \bz_i^*, 
\]
which is another way to express that $\cA^*\cA(\bX) = \sum_i \<\bz_i
\bz_i^*, \bX\> \bz_i \bz_i^*$.  Now observe the simple identity:
\begin{equation}
  \label{eq:EAtA}
  \E  [\bz_i \bz_i^* \otimes \bz_i \bz_i^*] = 
  2 \mathcal{I} + \bI_n \otimes \bI_n := \mathcal{S},
\end{equation}
where $\mathcal{I}$ is the identity operator and $\bI_n$ the
$n$-dimensional identity matrix.  Put differently, this means that for
all $\bX$,
\[
\cS(\bX) = 2\bX + \tr(\bX) \bI. 
\]
The proof is a simple calculation and omitted. It is also not hard to
see that the mapping $\cS$ is invertible and its inverse is given by
\[
\cS^{-1} = \frac12\Bigl(\mathcal{I} - \frac{1}{n+2} \bI_n \otimes \bI_n\Bigr)
\quad \Leftrightarrow \quad \cS^{-1}(\bX) = \frac12 \Bigl(\bX -
\frac{1}{n+2} \tr(\bX) \bI_n\Bigr). 
\]
We will use this object in the definition of our dual certificate. 

\subsection{Construction}

For pedagogical reasons, we first introduce a possible candidate
certificate defined by
\begin{equation}
  \label{eq:certificate}
  \bar{\bY} := \frac{1}{m} \cA^* \cA \cS^{-1}(\be_1 \be_1^*). 
\end{equation}
Clearly, $\bar{\bY}$ is in the range of $\cA^*$ as required. To justify this
choice, the law of large numbers gives that in the limit of infinitely
many samples, 
\[
\lim_{m \goto \infty} \frac{1}{m} \sum_i (\bz_i \bz_i^* \otimes \bz_i \bz_i^*)
\cS^{-1}(\be_1 \be_1^*) = \E (\bz_i \bz_i^* \otimes \bz_i \bz_i^*) \cS^{-1}(\be_1
\be_1^*) = \be_1 \be_1^*.
\]
In other words, in the limit of large samples, we have a perfect
certificate since $\bar{\bY}_T = \be_1 \be_1^*$ and $\bar{\bY}_\Tp =
0$. Our hope is that the sample average is sufficiently close to the
population average so that one can check \eqref{eq:dualcertif}. In
order to show that this is the case, it will be useful to think of
$\bar{\bY}$ \eqref{eq:certificate} as the random sum
\[
\bar{\bY} = \frac{1}{m} \sum_i \bY_i, 
\]
where each matrix $\bY_i$ is an independent copy of the random matrix 
\[
\frac12 \Bigl[z_1^2 - \frac{1}{n+2} \|\bz\|_2^2\Bigr] \bz \bz^*
\]
in which $\bz = (z_1, \ldots, z_n) \sim \mathcal{N}(0,I)$.

We would like to make an important point before continuing. We have
seen that all we need from $\bar{\bY}$ is
\[
\|\bar{\bY}_T - \be_1 \be_1^*\|_2 \le 1/3
\] 
(and $\|\bar{\bY}_\Tp\|\le 1/2$).  This is in stark contrast with David Gross'
approach \cite{Gross09} which requires a very small misfit, i.e.~an
error of at most $1/n^2$. In turn, this loose bound has an enormous
implication: it eliminates the need for the golfing scheme and allows
for the simple certificate candidate \eqref{eq:certificate}.  In fact,
our certificate can be seen as the first iteration of Gross' golfing
scheme.

\subsection{Truncation}

For technical reasons, it is easier to work with a truncated version
of $\bar{\bY}$ and our dual certificate is taken to be
\begin{equation}
\label{eq:truncated}
{\bY} = \frac{1}{m} \sum_i \bY_i \, 1_{E_i}, 
\end{equation}
where the $\bY_i$'s are as before and $1_{E_i}$ are independent copies
of $1_E$ with 
\[
E = \{|z_1| \le \sqrt{2\beta \log n}\} \, \cap \{\|\bz\|_2 \le
\sqrt{3n}\}.
\]
We shall work with $\beta = 3$ so that $|z_1| \le \sqrt{6 \log n}$.

\begin{lemma}
\label{lem:YT}
Let ${\bY}$ be as in~\eqref{eq:truncated}. Then
\begin{equation}
\label{YTbound}
\P \Bigl( \| {\bY}_T - \be_1 \be_1^*\|_2 
\ge \frac{1}{3} \Bigr) \le 2\exp \Bigl(-\gamma \frac{m}{n}\Bigr),  
\end{equation}
where $\gamma >0$ is an absolute constant. This holds with the proviso
that $m \ge c_1 \, n$ for some numerical constant $c_1 > 0$, and that
$n$ is sufficiently large.
\end{lemma}

\begin{lemma}
\label{lem:YTp}
Let ${\bY}$ be as in~\eqref{eq:truncated}. Then
\begin{equation}
\label{YTpbound}
\P \Bigl( \| {\bY}_\Tp \| \ge \frac{1}{2} \Bigr) \le 
4 \exp \Big(-\gamma \frac{m}{\log n}\Big). 
\end{equation}
where $\gamma > 0$ is an absolute constant. This holds with the
proviso that $m \ge c_1 \, n \log n$ for some numerical constant $c_1
> 0$, and that $n$ is sufficiently large.
\end{lemma}

\subsection{${\bY}$ on $T$ and proof of Lemma \ref{lem:YT}}

It is obvious that for any symmetric matrix $\bX \in T$,
\[
\|\bX\|_2 \le \sqrt{2} \|\bX \be_1\|_2 
\]
since only the first row and column are nonzero. We have 
\begin{equation}
\label{eq:split}
{\bY}_T \be_1 - \be_1 = \frac{1}{m} \sum_{i = 1}^m \by_i 1_{E_i} -
\frac{1}{m} \sum_{i=1}^m \be_1 \, 1_{E_i^c},
\end{equation}
where the $\by_i$'s are independent copies of the random vector
\begin{equation}
\label{eq:Xvector1}
\by = \frac12  \Bigl[z_1^2 - \frac{1}{n+2} \|\bz\|_2^2\Bigr] z_1 \,  \bz  -
\be_1 := (\xi z_1)\, \bz   - \be_1.  
\end{equation}

We claim that 
\[
\Bigr\Vert\frac{1}{m} \sum_{i=1}^m \be_1 \, 1_{E_i^c}\Bigr\Vert_2 \le 1/9,
\]
with probability at least $1 - 2 e^{-\gamma m}$ for some $\gamma >
0$. This is a simple application of Bernstein's inequality. Set
$\pi(\beta) = \P(E_i^c)$ and observe that 
\begin{equation}
\label{eq:truncation1}
\pi(\beta) = \P(|z_1| \ge \sqrt{2\beta \log n}) +
\P(\|\bz\|^2_2 \ge 3n) \le n^{-\beta} +
e^{-\frac{n}{3}}.  
\end{equation}
The right-hand side follows from $\P(|z_1| \ge t) \le e^{-t^2/2}$
which holds for $t \ge 1$ and from $\P(\|\bz\|_2^2 \ge 3n) \le
e^{-n/3}$. In turn, this last bound follows from
\[
\P(\|\bz\|^2_2 - n \ge \sqrt{2n} t + t^2) \le e^{-t^2/2}.
\]
Returning to Bernstein, this gives 
\[
\P\Bigl(\Bigl\vert\frac{1}{m} \sum_{i = 1}^m  1_{E_i^c} - \pi(\beta)
\Bigr\vert \ge t) \le 2\exp\Bigl(-\frac{mt^2}{2\pi(\beta) + 2 t/3}\Bigr). 
\]
Setting $t = 1/18$, $\beta = 3$ and taking $n$ large enough so that
$\pi(3) \le 1/18$ proves the claim.

The main task is to bound the $2$-norm of the sum $\sum_{i = 1}^m
\by_i 1_{E_i}$ and a convenient way to do this is via the vector
Bernstein inequality.
\begin{theorem}[Vector Bernstein inequality] 
  Let $\bx_i$ be a sequence of independent random vectors and set $V
  \ge \sum_i \E \|\bx_i\|_2^2$. Then for all $t \le V/\text{\em max}
  \|\bx_i\|_2$, we have
\[
\P(\|\sum_i (\bx_i - \E \bx_i)\|_2 \ge \sqrt{V} + t) \le e^{-t^2/4V}.
\]
\end{theorem}
It is because this inequality requires bounded random vectors that we
work with the truncation $\sum_{i = 1}^m \by_i 1_{E_i}$.

Put $\bar \by = \by \, 1_E$.  Since $\|\bar \by\|_2^2 \le
\|\by\|_2^2$, we first compute $\E \|\by\|_2^2$. We have
\[
\|\by\|_{2}^2 = \|\bz\|_{2}^2 z_1^2 \xi^2 - 2 z_1^2 \xi +1, \qquad \xi
= \frac12 \Bigl[z_1^2 - \frac{1}{n+2} \|\bz\|_2^2\Bigr], 
\]
and a little bit of algebra yields
\[
\|\by\|_{2}^2 =  \frac14 z_{1}^6 \|\bz\|_{2}^2 - \frac{1}{2(n+2)} z_{1}^4
\|\bz\|_{2}^4
+ \frac{1}{4(n+2)^2}z_{1}^{2} \|\bz\|_{2}^6 - z_{1}^4 +
\frac{1}{n+2}z_{1}^2 \|\bz\|_{2}^2 +1.
\]
Thus, 
\begin{align}
  \E\left[\|\by\|_{2}^2 \right] & = \frac14 (15n+90) -
  \frac{1}{2(n+2)}(3n^2 + 30n +72) + \frac{1}{4(n+2)}(n+4)(n+6) -1
  \nonumber \\ & \leq 4(n + 4), \label{eq:var1}
\end{align}
where we have used the following identities
\begin{align*}
\E\left[z_{1}^2 \|\bz\|_{2}^2 \right] & = n+2,\\
 \E\left[z_{1}^2 \|\bz\|_{2}^6 \right] & = (n+2)(n+4)(n+6),\\
  \E\left[z_{1}^4 \|\bz\|_{2}^4 \right] & = 3n^2 +30n +72,\\
   \E\left[z_{1}^6 \|\bz\|_{2}^2 \right] & = 15n+90.
\end{align*}

Second, on the event of interest we have $|\xi| \le \beta \log n$
(assuming $2\beta \log n \ge 3$), $|z_1| \le \sqrt{2\beta \log n}$ and
$\|\bz\|_2 \le \sqrt{3n}$ and, therefore,
\[
\|\bar \by\|_2 \le \sqrt{6n} \, (\beta \log n)^{3/2} + 1 \le \sqrt{7n}
(\beta \log n)^{3/2}
\] 
provided $n$ is large enough. 

Third, observe that by symmetry, all the entries of $\bar \by$ but the
first have mean zero. Hence,
\[
\|\E \bar{\by} \|_2 = |\E y_1-\bar{y}_1 | =
|\E\, 1_{E^c} y_1| \le \sqrt{\mathbb{P}\left( E^c \right)}
\, \sqrt{\E y_{1}^2}.
\]
We have
\[
y_{1}^2 = (\xi z_{1}^2 -1)^2 = \frac14 z_{1}^8 - z_{1}^4
+\frac{1}{n+2}\|z\|_{2}^2 z_{1}^2 - \frac{1}{2(n+2)}\|z\|_{2}^2
z_{1}^6 + \frac{1}{4(n+2)^2}\|z\|_{2}^4 z_{1}^4 + 1
\]
and using the identities above
\[
\E y_{1}^2  = \frac{101}{4} - \frac{27n^2 + 210n
  +288}{4(n+2)^2} \leq 22, 
\]
which gives
\[
\|\E \bar{\by} \|_2 \leq \sqrt{22(n^{-\beta} +
  e^{-\frac{n}{3}})}.
\]

Finally, with $V = 4m(n+4)$, Bernstein's inequality gives that for
each $t \le 4(n+4)/[\sqrt{7n} (\beta \log n)^{3/2}]$,
\[
\|m^{-1} \sum_i (\bar \by_i - \E \bar \by_i)\|_2 \ge 2\sqrt{\frac{n+4}{m}} + t 
\]
with probability at most $\exp\bigl(- \frac{mt^2}{16(n+4)}\bigr)$. It
follows that
\[
\|m^{-1} \sum_i \bar \by_i \|_2 \ge \sqrt{22(n^{-\beta} +
  e^{-\frac{n}{3}})} + 2\sqrt{\frac{n+4}{m}} + t
\]
with at most the same probability.  Our result follows by taking $t =
1/6$, $\beta = 3$, $m \ge c_1 n$ where $n$ and $c_1$ are sufficiently
large such that
\[
\sqrt{22(n^{-\beta} +
  e^{-\frac{n}{3}})} + 2\sqrt{\frac{n+4}{m}} + \frac{1}{6} \le \frac{2}{9}. 
\]



\subsection{${\bY}$ on $\Tp$ and proof of Lemma \ref{lem:YTp}}

We have 
\[
{\bY}_\Tp = \frac{1}{m} \sum_i \bX_i \, 1_{E_i},
\]
where the $\bX_i$'s are independent copies of the random matrix
\begin{equation}
\label{eq:Xvector2}
\bX = \frac12 \Bigl[z_1^2 - \frac{1}{n+2} \|\bz\|_2^2\Bigr] 
\, \PTp(\bz \bz^T). 
\end{equation}
One natural way to bound the norm of this random sum is via the
operator Bernstein's inequality. We develop a more customized
approach, which gives sharper results.

Decompose $\bX$ as
\[
\bX = \frac12 \Bigl[z_1^2 - 1\Bigr] 
\, \PTp(\bz \bz^T) + \frac12 \Bigl[1 - \frac{1}{n+2} \|\bz\|_2^2\Bigr] 
\, \PTp(\bz \bz^T) := \bX^{(0)} + \bX^{(1)}.
\]
Note that since $z_1$ and $\PTp(\bz\bz^T)$ are independent, we have
$\E \bX^{(0)} = 0$ and thus, $\E \bX^{(1)} = 0$ since $\E \bX = 0$.
With $\bar \bX_i^{(0)} = \bX_i^{(0)} 1_{E_i}$ and similarly for $\bar
\bX_i^{(1)}$, it then suffices to show that
\begin{equation}
\label{eq:X0X1}
\Bigl\Vert\sum_i \bar{\bX}_i^{(0)}\Bigr\Vert \le m/4 \quad \text{and} \quad \Big\Vert\sum_i
\bar{\bX}_i^{(1)}\Big\Vert \le m/4
\end{equation}
with large probability.  Write the norm as
\[
\Bigl\Vert\sum_i \bar{\bX}_i^{(0)}\Bigr\Vert = \sup_{\bu} \,
\Bigl\vert \sum_i \<\bu, \bar{\bX}_i^{(0)} \bu\>\Bigr\vert,
\]
where the supremum is over all unit vectors $\bu$ that are orthogonal to
$\be_1$.  The strategy is now to find a bound on the right-hand side for
each fixed $\bu$ and apply a covering argument to control the supremum
over the whole unit sphere.  In order to do this, we shall make use of
a classical large deviation result.
\begin{theorem}[Bernstein inequality] \label{teo:Bernstein} Let
  $\{X_i\}$ be a finite sequence of independent random variables.
  Suppose that there exists $V$ and $c$ such that for all $k \ge 3$,
  \[
  \sum_i \E |X_i|^k \le \frac12 k! V c_0^{k-2}.
  \] 
  Then for all $t \geq 0$,
\begin{equation}
\label{eq:bernstein}
\P\Bigl(\Bigl|\sum_i X_i -\E X_i\Bigr| \geq t\Bigr) \leq 2  
\exp\Bigl(-\frac{t^2}{2V + 2c_0t}\Bigr).
\end{equation}
\end{theorem}

For the first sum in \eqref{eq:X0X1}, we write
\[
\sum_i \<\bu, \bar{\bX}_i^{(0)} \bu\> = \sum_i \eta_i \, 1_{E_i}, 
\]
where the $\eta_i$'s are independent copies of 
\[
\eta = \frac12 \Bigl[z_1^2 - 1\Bigr] \<\bz, \bu\>^2.
\]
The point of the decomposition $\bX^{(0)} + \bX^{(1)}$ is that $z_1$ and
$\<\bz,\bu\>$ are independent since $\bu$ is orthogonal to $\be_1$.   We have
$\E \eta = 0$ and for $k \ge 2$,
\[
\E |\eta \, 1_{E}|^k \le 2^{-k} \E |(z_1^2 - 1) \, 1_{\{z_1^2 \le
  2\beta \log n\}}|^k \, \E |\<\bz,\bu\>|^{2k}.
\]
First, 
\[
\E |(z_1^2 - 1) \, 1_{\{z_1^2 \le 2\beta \log n\}}|^k \le (2\beta \log
n)^{k-2} \E (z_1^2 -1)^2 = 2 (2\beta \log n)^{k-2}.
\]
Second, the moments of a chi-square variable with one degree of
freedom are well known: 
\[
\E |\<\bz,\bu\>|^{2k} = 1 \times 3 \times \ldots \times (2k-1) \le 2^k k!
\]
Hence we can apply Bernstein inequality with $V = 4m$ and $c_0 =
2\beta \log n$ and, obtain
\[
\P\Bigl(\Bigl\vert\sum_i \eta_i \, 1_{E_i} - \E [\eta_i \,
1_{E_i}]\Bigr\vert \ge mt\Bigr) \le 2 \exp\Bigl(-\frac{m}{4}
\frac{t^2}{2 + \beta t \log n}\Bigr).
\]
We now note that
\[
\vert \E \eta_i 1_{E_i}\vert = \vert \E \eta_i 1_{E_i^c}\vert \le
\sqrt{\P(E_i^c)} \sqrt{\E \eta_i^2} = \sqrt{\frac{3\pi(\beta)}{2}}
\]
which gives
\[
\P\Bigl(m^{-1}\Bigl\vert\sum_i \eta_i \, 1_{E_i} \Bigr\vert \ge t +
\sqrt{\frac{3\pi(\beta)}{2}}\Bigr) \le 2 \exp\Bigl(-\frac{m}{4}
\frac{t^2}{2 + \beta t \log n}\Bigr).
\]
For instance, take $t = 1/12$, $\beta = 3$, $m \geq c_1 n$ and $n$
large enough to get
\[
\P\Bigl(m^{-1} \Bigl\vert\sum_i \eta_i \, 1_{E_i} \Bigr\vert \ge 1/8 \Bigr) \le 2 \exp\Bigl(-\gamma \frac{m}{\log n}\Bigr).
\]

To derive a bound about $\|\bar{\bX}^{(0)}\|$, we use (see Lemma 4 in
\cite{VershyninRMT})
\[
\sup_u \Bigl\vert \<\bu, \bar{\bX}^{(0)} u\>\Bigr\vert \le 2 \sup_{\bu \in
  \mathcal{N}_{1/4}} \Bigl\vert \<\bu, \bar{\bX}^{(0)} \bu\>\Bigr\vert, 
\]
where $ \mathcal{N}_{1/4}$ is a $1/4$-net of the unit sphere $\{\bu :
\|\bu\|_2 = 1, \bu \perp \be_1\}$.  Since $|\mathcal{N}_{1/4}| \le
9^n$,
\[
\P(m^{-1}\|\bar{\bX}^{(0)}\| > 1/4) \le \P\Bigl(m^{-1} \sup_{\bu \in
  \mathcal{N}_{1/4}} \Bigl\vert \<\bu, \bar{\bX}^{(0)} \bu\>
\Bigr\vert > 1/8 \Bigr) \le 9^n \times 2\exp\Bigl(-\gamma
\frac{m}{\log n}\Bigr).
\]

We deal with the second term in a similar way, and write 
\[
\sum_i \<\bu, \bar{\bX}_i^{(1)} \bu\> = \sum_i \eta_i \, 1_{E_i}, 
\]
where the $\eta_i$'s are now independent copies of
\[
\eta = \frac12 \Bigl[1 - \frac{\|\bz\|^2_2}{n+2}\Bigr] \<\bz, \bu\>^2.
\]
On $E$, $\|\bz\|_2^2 \le 3n$ and, therefore, $\E |\eta \, 1_E|^k \le
2^{k} k!$. We can apply Bernstein's inequality with $c_0 = 2$ and $V =
8m$, which gives
\[
\P\Bigl(\Bigl\vert\sum_i \eta_i \, 1_{E_i} - \E [\eta_i \,
1_{E_i}]\Bigr\vert \ge mt\Bigr) \le 2 \exp\Bigl(-\frac{m}{4} \frac{t^2}{4 +
   t}\Bigr).
\]
The remainder of the proof is identical to that above and is therefore omitted.

\subsection{Proof of Theorem~\ref{teo:main}} \label{subsec:proofoftheorem}

We now assemble the various intermediate results to establish 
Theorem~\ref{teo:main}. As pointed out, Theorem~\ref{teo:main} follows
immediately from Lemma~\ref{lem:crucial}, which in turn hinges on
the validity of the conditions stated in~\eqref{eq:RIP1r1},
\eqref{eq:RIP1r2}, and~\eqref{eq:dualcertif}.

Lemma~\ref{lem:A1} asserts that condition~\eqref{eq:RIP1r1} holds with
probability of failure at most $p_1$, where $p_1 = 2e^{-\gamma_1 m}$
and here and below, $\gamma_{1}, \ldots, \gamma_4$ are positive
numerical constants.  Similarly, Lemma~\ref{lem:A2} shows that
condition~\eqref{eq:RIP1r2} holds with probability of failure at most
$p_2$, where $p_2 = 3e^{-\gamma_2 m}$.  In both cases we need that $m
> cn$ for an absolute constant $c>0$.

Proceeding to the dual certificate in~\eqref{eq:dualcertif}, we note
that Lemma~\ref{lem:YT} establishes the first part of the dual
certificate with a probability of failure at most $p_3$, where $p_3 =
3e^{-\gamma_3 m/n}$. The second part of the dual certificate
in~\eqref{eq:dualcertif} is shown in Lemma~\ref{lem:YTp} to hold with
probability of failure at most $p_4$, where $p_4 = 4e^{-\gamma_4
  \frac{m}{\log n}}$. In the former case we need $m > cn$ for an
absolute constant $c > 0$ and in the latter $m > c' n \log n$.

Finally, the union bound gives that under the hypotheses of
Theorem~\ref{teo:main}, exact recovery holds with probability at least
$1-3e^{-\gamma m/n}$ for some $\gamma > 0$, as claimed.


\section{The Complex Model} \label{sec:complex}

This section proves that Theorem~\ref{teo:main} holds for the complex
model as well.  Not surprisingly, the main steps of the proof are the
same as in the real case, but there are here and there some noteworthy
differences.  Instead of deriving the whole proof, we will carefully
indicate the nontrivial changes that need to be carried out.

First, we can work with $\bx = \be_1$ because of rotational invariance,
and with independent complex valued Gaussian sequences $\bz_i \sim
\mathcal{C}\mathcal{N}(0,I,0)$. This means that the real and imaginary
parts of $\bz_i$ are independent white noise sequences with variance
$1/2$.

The key Lemma \ref{lem:crucial} only requires a slight adjustment in
the numerical constants. The reason for this is that while Lemma
\ref{lem:A1} does not require any modification, Lemma \ref{lem:A2}
changes slightly; in particular, the numerical constants are somewhat
different. Here is the properly adjusted complex version.
\begin{lemma}
  \label{lem:A2C}
  Fix $\delta > 0$. Then there are positive numerical constants $c_0$
  and $\gamma_0$ such that if $m \ge c_0 \, [\delta^{-2} \log
  \delta^{-1}] \, n$, $\cA$ has the following property with
  probability at least $1 - 3 e^{-\gamma_0 m\delta^2}$: for any
  Hermitian rank-2 matrix $X$,
\begin{equation}
\label{eq:A2bC}
\frac{1}{m} \|\cA(\bX)\|_1 \ge 2(\sqrt{2}-1) (1-\delta) \|\bX\| \ge
0.828(1-\delta)\|\bX\|.
\end{equation}
\end{lemma}
The proof of this lemma follows essentially the proof of
Lemma~\ref{lem:A2}. 
The function $f(t)$ (cf.\ equation~\eqref{eq:f}) 
now takes the form
\begin{equation}
\E \xi = f(t) =  \frac{1+t^2}{1+t},
\label{eq:fcomplex}
\end{equation}
where $\xi = \big| |Z_1|^2 - t |Z_2|^2\big|$, with $Z_1$ and $Z_2$
independent $\mathcal{C}\mathcal{N}(0,1,0)$, as demonstrated in the
following lemma.

\begin{lemma}
  Let $Z_1$ and $Z_2$ be independent $\mathcal{C}\mathcal{N}(0,1,0)$ variables and
  $t \in [0,1]$. We have
\[
E ||Z_1|^2 - t |Z_2|^2| = f(t), 
\] 
where $f(t)$ is given by \eqref{eq:fcomplex}.
\end{lemma}
\begin{proof}
  Set
\[
\rho = \frac{1-t}{1+t} \text{ and } \cos
\theta = \rho
\]
in which $\theta \in [0, \pi/2]$.  By using polar coordinates for the variables $(x_1,y_1)$ associated with $Z_1$ and $(x_2,y_2)$, associated with $Z_2$ we have
\begin{align*}
  \E ||Z_1|^2 - t |Z_2|^2| & = \frac{1}{2} \int_{0}^\infty  \, \int_0^{\infty} |r_{1}^2 -tr_{2}^2|r_1 r_2 e^{-r_{1}^{2}/2}e^{-r_{2}^{2}/2} \, dr_1 dr_2\\
& =   \frac{1}{8} \, \int_0^{\infty} r^5 e^{-r^{2}/2}\, dr \int_{0}^{2\pi}
|\sin\phi\cos\phi|| \cos^2 \phi - t \sin^2\phi| \, d\phi,
\end{align*}
where we used polar coordinates again in variables $(r_1,r_2)$. Now using the identities $\cos^2 \phi = (1+\cos 2\phi)/2$,
$\sin^2\phi = (1-\cos 2\phi)/2$ and $2\sin\phi\cos\phi = \sin2\phi$ we have
\begin{align*}
  \E |Z_1^2 - t Z_2^2| & =  \frac12 \, \int_0^\pi |\sin2\phi||\cos 2\phi +
  \rho| \, d\phi\\
  & = \frac{1}{2} \,\Bigl[ \int_0^\theta \sin\phi(\cos \phi - \rho) \, d\phi
  + \int_\theta^\pi \sin\phi(\rho - \cos \phi) \, d\phi \Bigr]\\
  & = \frac{1}{2} (1+t) [-\frac{1}{2}\cos2\theta +2\rho\cos\theta + \frac12]\\
  & = \frac12 (1+t)[\rho^2+1]\\
  & = \frac{1+t^2}{1+t}
\end{align*}
as claimed. 
\end{proof}

The graph of $f(t)$ is shown in Figure \ref{fig:fcomplex}.
\begin{figure}
\centering
    \includegraphics[width=0.49\textwidth]{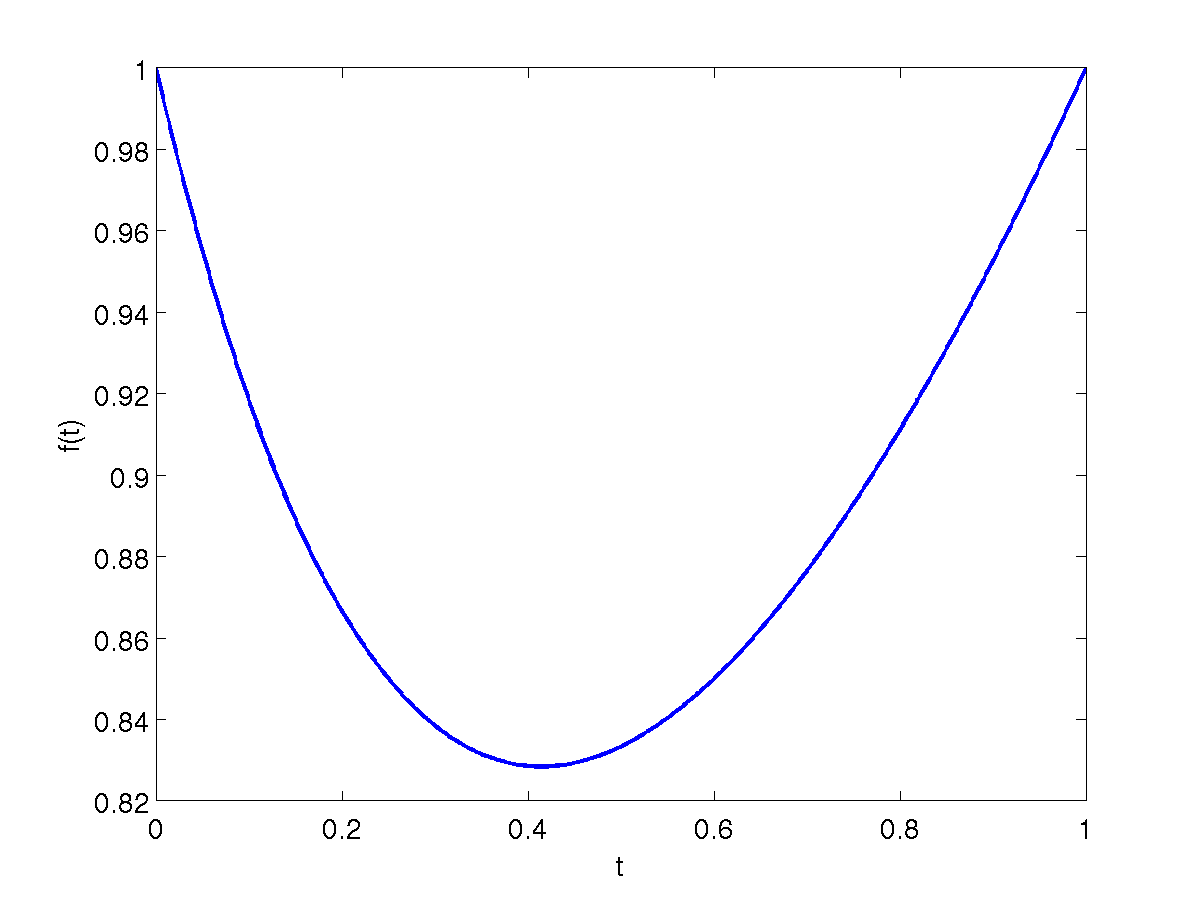}
    \caption{The function $f(t)$ in \eqref{eq:fcomplex} as a function
      of $t$.}
\label{fig:fcomplex}
\end{figure}
The minimum of this function on $\left[0,1\right]$ is $2(\sqrt{2}-1) >
0.828$. Furthermore, the covering argument in that proof has to be
adapted; for example, unit spheres need to be replaced by complex unit
spheres.

A consequence of this change in numerical values is that the numerical
factors in Lemma \ref{lem:RIP} need to be adjusted.
\begin{lemma}
  \label{lem:RIPcomplex}
  Any feasible matrix $\bH$ such that $\tr(\bH) \le 0$ must obey
\[
\|\bH_T\|_2 \le \sqrt{\frac{5}{4}} \, \|\bH_T\|.
\]
\end{lemma}

Finally, with all of this in place, Lemma~\ref{lem:crucial} becomes this: 
\begin{lemma}
  \label{lem:crucialcomplex}
  Suppose that the mapping $\cA$ obeys the following two properties:
  for some $\delta \le 3/13$: 1) for all positive semidefinite
  matrices $\bX$,
\begin{equation}
\label{eq:RIP1r1C}
m^{-1} \|\cA(\bX)\|_1 \leq (1+\delta) \|\bX\|_{1};
\end{equation}
2) for all matrices $\bX \in T$
\begin{equation}
\label{eq:RIP1r2C}
m^{-1} \|\cA(\bX)\|_1 \geq 2(\sqrt{2}-1)(1-\delta)\|\bX\| \geq
0.828(1-\delta) \|\bX\|.
\end{equation}
Suppose further that there exists $Y$ in the range of $\cA^*$ obeying
\begin{equation}
  \label{eq:dualcertifC}
  \|\bY_T - \be_1 \be_1^*\|_2 \le 1/5 \quad \text{and} \quad \|\bY_\Tp\| \le 1/2.
\end{equation}
Then $\be_1 \be_1^*$ is the unique minimizer to \eqref{eq:tracemin}.
\end{lemma}

We now turn our attention to the properties of the dual certificate we
studied in Section \ref{sec:dual}. The first difference is that the
expectation of $\cA^* \cA$ in \eqref{eq:EAtA} is different in the
complex case. A simple calculation yields
\[
\E \frac{1}{m} \cA^* \cA =  \mathcal{I} + I_n \otimes I_n := \mathcal{S}.
\]
This means that for all $\bX$,
\begin{equation}
\label{eq:SC}
\cS(\bX) = \bX + \tr(\bX) \bI.
\end{equation}
We note that in this case
\begin{equation}
\label{eq:SinvC}
\cS^{-1} = \mathcal{I} - \frac{1}{n+1} \bI_n \otimes \bI_n
\quad \Leftrightarrow \quad \cS^{-1}(\bX) =  \bX -
\frac{1}{n+1} \tr(\bX) \bI_n.
\end{equation}
We of course use this new $\cS^{-1}$ in the complex analog of the
candidate certificate \eqref{eq:truncated}.  A consequence is that
in the proof of Lemma~\ref{lem:YT}, for instance, \eqref{eq:Xvector1}
now takes the form
\begin{equation}
\label{eqLXvector1C}
\bX =   \Bigl[|z_1|^2 - \frac{1}{n+1} \|\bz\|_2^2\Bigr] \bar{z}_1 \,  \bz -
\be_1
:= (\xi \bar{z}_1)\, \bz - \be_1.
\end{equation}
To bound the 2-norm of a sum of i.i.d.\ such random variables (as in Lemma~\ref{lem:YT}), we employ the same Bernstein inequality for real vectors, 
using the fact that $\|\bz\|_2 = \|(\Re(\bz),\Im(\bz))\|_2$ for any complex
vector $\bz$.
Similarly \eqref{eq:Xvector2} becomes
\begin{equation}
\label{eq:Xvector2C}
\bX =  \Bigl[|z_1|^2 - \frac{1}{n+1} \|\bz\|_2^2\Bigr] \, \PTp(\bz \bz^*).
\end{equation}
To bound the operator norm of a sum of i.i.d.\ such random matrices (as in Lemma~\ref{lem:YTp}), 
we again use a covering argument, this time working with chi-square variables with two
degrees of freedom, since $|\langle \bz,\bu \rangle|^2$ is distributed as
$\frac{1}{2}\chi^2(2)$.  Since $|\langle \bz,\bu \rangle|^2$ are real random variables,
we use the same version of the Bernstein inequality as in the real-valued case.  The only difference is that 
the moments  are now
$$
\E |\<\bz,\bu\>|^{2k} = 2^{-k} \times (2+0) \times (2+2) \times (2+4) \times
\ldots \times (2 + 2k-2) = k!
$$

\section{Stability}
\label{sec:stability}

This section proves the stability of our approach, namely, Theorem
\ref{teo:stability}. Our proof parallels the argument of Cand\`es and
Plan for showing the stability of matrix completion \cite{MCNoise} as
well as that of Gross et al. in \cite{QSTC}.

Just as before, we prove the theorem in the real case since the
complex case is essentially the same. Further, we may still take $\bx
=\be_1$ without loss of generality.  We shall prove stability when the
$\bz_i$'s are i.i.d.\ $\mathcal{N}(0,\bI_n)$ and later explain how one can
easily transfer a result for Gaussian vectors to a result for vectors
sampled on the sphere.
Under the assumptions of the theorem, the RIP-1-like
properties, namely, Lemmas \ref{lem:A1} and \ref{lem:A2} hold with a
numerical constant $\delta_1$ we shall specify later.  Under the same
hypotheses, the dual certificate $\bY$ \eqref{eq:certificate} obeys
\[
\| \PT(\bY-\be_1 \be_1^*)\|_2 \leq \gamma, \qquad \|\bY_{\Tp}\| \leq \frac12, 
\]
in which $\gamma$ is a numerical constant also specified later.

Set $\bX = \bx\bx^* = \be_1 \be_1^*$ and write $\hat \bX = \bX + \bH$.
We begin by recording two useful properties. First, since $\bX$ is
feasible for our optimization problem, we have
\begin{equation}
  \label{eq:cone}
  \tr(\bX+\bH) \le \tr(\bX) \quad \Longleftrightarrow \quad \tr(\bH) \le 0. 
\end{equation}
Second, the triangle inequality gives 
\begin{equation}
  \label{eq:tube}
  \|\cA(\bH)\|_2 = \|\cA(\hat \bX - \bX)\|_2 \le  \|\cA(\hat \bX) - \bb\|_2
+ \|\bb - 
\cA(\bX)\|_2 \le 2\epsilon.
\end{equation}
In the noiseless case, $\cA(\bH) = 0 \Longrightarrow \<\bH, \bY\> = 0$, by
construction. In the noisy case, a third property is that $|\<\bH, \bY\>|$
is at most on the order of $\epsilon$. Indeed,
\[
m |\<\bH,\bY\>| = |\<\cA(\bH), \cA \cS^{-1}(\bX)\>| \le \|\cA(\bH)\|_\infty \|\cA
\cS^{-1}(\bX)\|_1. 
\]
Since, $\|\cA(\bH)\|_\infty \le \|\cA(\bH)\|_2$ and 
\[
\|\cA \cS^{-1}(\bX)\|_1 \le m(1+\delta_1) \|\cS^{-1}(\bX)\|_1  \leq m(1+\delta_1), 
\]
we obtain
\begin{equation}
  \label{eq:HY}
  |\<\bH,\bY\>| \le 2\epsilon(1+\delta_1). 
\end{equation}

We now reproduce the steps of the proof of Lemma \ref{lem:crucial},
and obtain
\[
0 \ge \tr(\bH_T) + \tr(\bH_\Tp) \ge \frac12 \tr(\bH_\Tp) - \gamma \|\bH_T\|_2
- |\<\bH,\bY\>|, 
\]
which gives 
\begin{equation}
  \label{eq:intermediate1}
  \tr(\bH_\Tp) \le 4\epsilon(1+\delta_1) + 2\gamma \|\bH_T\|_2 \le
4\epsilon(1+\delta_1) + 2\sqrt{2}\gamma \|\bH_T\|, 
\end{equation}
where we recall that $\bH_T$ has rank at most 2. We also have
\begin{align}
\nonumber 0.94(1-\delta_1) \|\bH_T\|  \le m^{-1} \|\cA(\bH_T)\|_1 & \le m^{-1}
\|\cA(\bH)\|_1 + m^{-1} \|\cA(\bH_\Tp)\|_1\\
& \le m^{-1/2}
\|\cA(\bH)\|_2 + (1+\delta_1) \tr(\bH_\Tp)\\
& \le 2m^{-1/2}\epsilon + (1+\delta_1)\tr (\bH_{\Tp}),
\label{eq:intermediate2} 
\end{align}
where the second inequality follows from the RIP-1 property together
with the Cauchy-Schwarz inequality. Plugging this last bound into
\eqref{eq:intermediate1} gives
\[
\tr(\bH_\Tp) \le 4\epsilon (1+\delta_1 + \gamma \alpha m^{-1/2}) + \beta
\gamma \tr(\bH_\Tp),
\]
where 
\[
\alpha = \frac{\sqrt{2}}{0.94(1-\delta_1)}, \quad \beta = 2\alpha (1+\delta_1). 
\]
Hence, when $\beta \gamma < 1$, we have 
\[
\tr(\bH_\Tp) = \|\bH_\Tp\|_1 \le\frac{4(1+\delta_1 + \gamma \alpha
m^{-1/2})}{1-\beta \gamma}\epsilon = c_1 \, \epsilon.
\]
 In addition, \eqref{eq:intermediate2} then gives
\[
\|\bH_T\| \le  \frac{2m^{-1/2}+(1+\delta_1)c_1}{0.94(1-\delta_1)}\epsilon =
c_2 \, \epsilon. 
\]
In conclusion, 
\[ \|\bH\|_2 \le \|\bH_T\|_2 + \|\bH_\Tp\|_2 \le \sqrt{2} \|\bH_T\| +
  \|\bH_\Tp\|_1 \le (\sqrt{2}c_2 + c_1)\epsilon = c_0 \, \epsilon, 
\]
and we also have $\|\bH\| \le (c_2 + c_1) \epsilon$.

It remains to show why the fact that $\hat \bX$ is close to $\bX =
\bx\bx^*$ in the Frobenius or operator norm produces a good estimate
of $\bx$ (recall that $\bx = \be_1$). Set $\epsilon_0 := \|\hat \bX -
\bX\| \le c_0\, \epsilon$. Below, $\hat{\lambda}_1 \ge 0$ is the largest
eigenvalue of $\hat \bX \succeq 0$, and $\hat{\bu}_1$ the first
eigenvector. Likewise, $\lambda_1 = 1$ is the top eigenvalue of $\bX =
\be_1 \be_1^*$. Since $\tr(\hat \bX) \le \tr(\bX)$,
\[
\hat{\lambda}_1 \le \lambda_1.
\] 
In the other direction, we know from perturbation theory that 
\[
|\lambda_1 - \hat \lambda_1| \le \|\hat \bX - \bX\| = \epsilon_0. 
\]
Assuming that $\epsilon_0 < 1$, this gives $\hat \lambda_1 \in
[1-\epsilon_0, 1]$. The sin-$\theta$-Theorem \cite{DavisKahan} implies
that
\[
|\sin \theta| \le \frac{\|\hat \bX - \bX\|}{|\hat \lambda_1|} \le
\frac{\epsilon_0}{1-\epsilon_0},
\]
where $ 0 \le \theta \le \pi/2$ is the angle between the spaces
spanned by $\hat{\bu}_1$ and $\be_1$. Writing 
\[
\hat{\bu}_1 = \cos\theta \be_1 + \sin \theta \be_1^\perp
\]
in which $\be_1^\perp$ is a unit vector orthogonal to $\be_1$,
Pythagoras' relationship gives
\[
\|\be_1 - \sqrt{\hlambda_1} \hat{\bu}_1\|_2^2 = (1-\sqrt{\hlambda_1}
\cos\theta)^2 + \hlambda_1 \sin^2\theta. 
\]
Since $\cos \theta = \sqrt{1-\sin^2\theta}$, we have
\[
1 \ge \sqrt{\hlambda_1} \cos \theta \ge \sqrt{1-\epsilon_0
  -\frac{\epsilon_0^2}{1-\epsilon_0}} \ge 1-\epsilon_0
\]
for $\epsilon_0 < 1/3$. Hence,
\[
\|\be_1 - \sqrt{\hlambda_1} \hat{\bu}_1\|_2^2 \leq \epsilon_0^2 +
\frac{\epsilon_0^2}{(1-\epsilon_0)^2} \le \frac{13}{4} \epsilon_0^2
\]
provided $\epsilon_0 < 1/3$. Since we always have
\[
\|\be_1 - \sqrt{\hlambda_1} \hat{\bu}_1\|_2 \le \|\be_1\|_2 +
\hlambda_1 \|\hat{\bu}_1\|_2 \le 2,
\]
we have established
\[
\|\be_1 - \sqrt{\hlambda_1} \hat{\bu}_1\|_2 \le C_0 \, \text{min}(\epsilon, 1). 
\]
This holds for all values of $\epsilon_0$ and proves the claim in the
case where $\|\bx\|_2 = 1$. The general case is obtained via a simple
rescaling.

\newcommand{\tbz}{\tilde{\bz}}
\newcommand{\tcA}{\tilde{\cA}}

As mentioned above, we proved the theorem for Gaussian $\bz_i$'s but
it is clear that our results hold true for vectors sampled uniformly
at random on the sphere of radius $\sqrt{n}$. The reason is that of
course, $\|\bz_i\|_2$ deviates very little from $\sqrt{n}$. Formally,
set $\tbz_i = [\sqrt{n}/\|\bz_i\|_2] \bz_{i}$ so that these new
vectors are independently and uniformly distributed on the sphere of
radius $\sqrt{n}$. Then
\[
\<\bX,\tbz_i \tbz_i^*\> = \frac{n}{\|\bz_i\|_2^2} \, \<\bX,\bz_i
\bz_i^*\>,
\]
and thus $\<\bX,\bz_i \bz_i^*\>$ is between $(1-\delta_2) \,
\<\bX,\tbz_i \tbz_i^*\>$ and $(1+\delta_2) \, \<\bX,\tbz_i \tbz_i^*\>$
with very high probability. This holds uniformly over all Hermitian
matrices. Thus if $\tcA(\bX) = \{\tbz_i^* X\tbz_i\}_{1\leq i \leq m}$,
\[
(1-\delta_2)\|\tcA(\bX)\|_{q} \leq \| \cA(\bX)\|_{q} \leq (1+\delta_2)\|\tcA(\bX)\|_{q} 
\]
for any $1 \le q \le \infty$.  

Now take $b_i = |\<\bx,\tbz_i\>|^2 + \nu_i$ and solve
\eqref{eq:traceminnoisy} to get $\tilde{\bX} = \bX +
\tilde{\bH}$. Going through the same steps as above by using the
relationships between $\cA$ and $\tcA$ throughout, and by using the
dual certificate $\bY$ associated with $\cA$, we obtain
\[
\|\tcA(\tilde{\bH})\|_2 \leq 2\epsilon, \qquad |\<\tilde{\bH},\bY \>|
\leq 2\epsilon(1+\delta_1)(1+\delta_2),
\]
and
\[
\tr(\tilde{\bH}_{\Tp}) \leq (1+\delta_2)c_1 \epsilon, \qquad \|\tilde{\bH}_{T}\| \leq
(1+\delta_2)c_2 \epsilon.
\]
Therefore, 
\[
\|\tilde{\bH}\|_2 \leq (1+\delta_2)(\sqrt{2}c_2+c_1)\epsilon.
\]
The rest of the proof goes through just the same.


\section{Numerical Simulations}
\label{sec:numerics}

In this section we illustrate our theoretical results with numerical
simulations. In particular, we will demonstrate PhaseLift's robustness
vis a vis additive noise.

We consider the setup in Section~\ref{subsec:stability}, where the
measurements are contaminated with additive noise.  The solution
to~\eqref{eq:traceminnoisy} is computed using the following
regularized nuclear-norm minimization problem:
\begin{equation}
\label{eq:lasso}
\text{minimize} \quad  \frac{1}{2} \|\cA(\bX) - \bb\|_2^2 + \lambda \|\bX\|_1.
\end{equation}
It follows from standard optimization theory~\cite{Roc70}
that~\eqref{eq:lasso} is equivalent to~\eqref{eq:traceminnoisy} for
some value of $\lambda$.  Hence, we use~\eqref{eq:lasso} to compute the
solution of~\eqref{eq:traceminnoisy} by determining via a simple and
efficient bisection search the largest value $\lambda(\epsilon)$ such
that $\|\cA(\bX) - \bb\|_2 \le \epsilon$.  The numerical algorithm to
solve~\eqref{eq:lasso} was implemented in Matlab using
TFOCS~\cite{BCG10}.  We then extract the largest rank-1 component as
described in Section~\ref{subsec:stability} to obtain an approximation
$\hat{\bx}$.

We will use the relative mean squared error (MSE) and the relative
root mean squared error (RMS) to measure performance. However, since a
solution is only unique up to global phase, it does not make sense to
compute the distance between $\bx$ and its approximation $\hat{\bx}$.
Instead we compute the distance modulo a global phase term and define
the relative MSE between $\bx$ and $\hat{\bx}$ as
\[
\min_{c : |c| = 1} \,\, \frac{\|c \bx - \hat{\bx}\|_2^2}{\|\bx\|_2^2}.
\]
The (relative) RMS is just the square root of the (relative) MSE.

In the first set of experiments, we investigate how the reconstruction
algorithm performs as the noise level increases.  The test signal is a
complex-valued signal of length $n=128$ with independent Gaussian
complex entries (each entry is of the form $a + i b$ where $a$ and $b$
are independent $\mathcal{N}(0,1)$ variables) so that the real and
imaginary parts are independent white noise sequences. Obviously, the
signal is arbitrary. We use $m=6n$ measurement vectors sampled
independently on the unit sphere $\C^n$.

We generate noisy data from both a Gaussian model and a Poisson
model. In the Gaussian model, $b_i \sim \mathcal{N}(\mu_i,\sigma^2)$
where $\mu_i = |\<\bx,\bz_i\>|^2$ and $\sigma$ is adjusted so that the
total noise power is bounded by $\epsilon^2$. In the Poisson model,
$b_i \sim \text{Poi}(\mu_i)$ and the noise $b_i - \mu_i$ is rescaled
to achieve a desired total power as above (we might do without this
rescaling as well but have decided to work with a prescribed
signal-to-noise ratio SNR for simplicity of exposition). We do this
for five different SNR levels,\footnote{The SNR of two signals $\bx,
  \hat{\bx}$ with respect to $\bx$ is defined as $10\log_{10}
  \|\bx\|_2^2/\|\bx - \hat{\bx}\|_2^2$. So we say that the SNR is 10dB
  if $10\log_{10} \|\bx\|_2^2/\|\bm{\nu}\|_2^2 = 10$.}  ranging from
5dB to 100dB.  However, we point out that we do not make use of the
noise statistics in our reconstruction algorithm\footnote{We refer
  to~\cite{CESV} for efficient ways to incorporate statistical noise
  models into the reconstruction algorithm.}, since our purpose is
only to assume an upper bound on the total noise power, as in
Theorem~\ref{teo:stability}.

  For each SNR level, we repeat the experiment ten times with
  different noise terms, different signals, and different random
  measurement vectors; we then record the average relative RMS over
  these ten experiments. Figure~\ref{fig:num1}(a) shows the average
  relative MSE in dB (the values of $10 \log_{10}(\text{rel.~MSE})$
  are plotted) versus the SNR for Poisson noise.  In each case, the
  performance degrades very gracefully with decreasing SNR, as
  predicted by Theorem~\ref{teo:stability}.  Debiasing as described at
  the end of Section~\ref{subsec:stability} leads to a further
  improvement in the reconstruction for low SNR, as illustrated in
  Figure~\ref{fig:num1}(b).  The results for Gaussian noise are
  comparable, see Figure~\ref{fig:num1a}.

\begin{figure}
\begin{center}
\subfigure[]{\includegraphics[width=75mm]{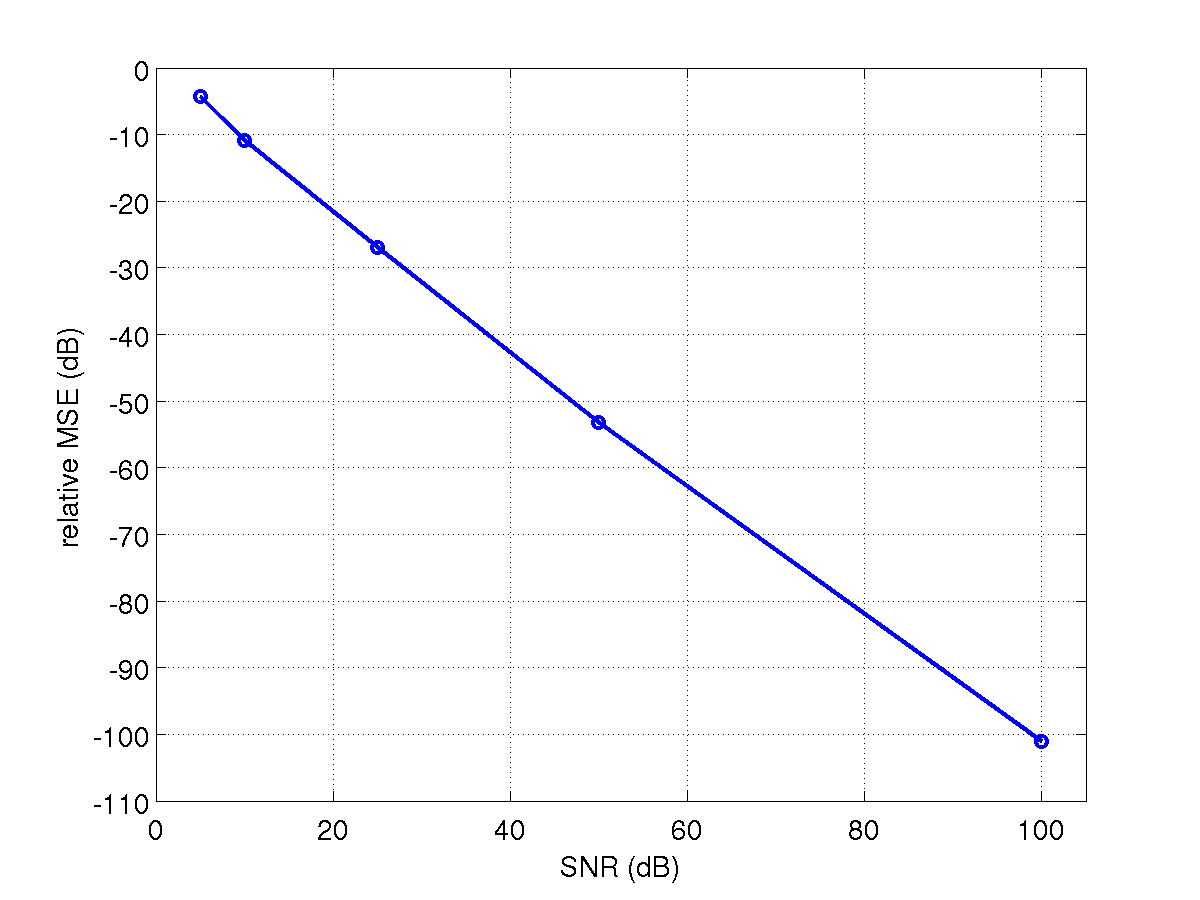}}
\quad
\subfigure[]{\includegraphics[width=75mm]{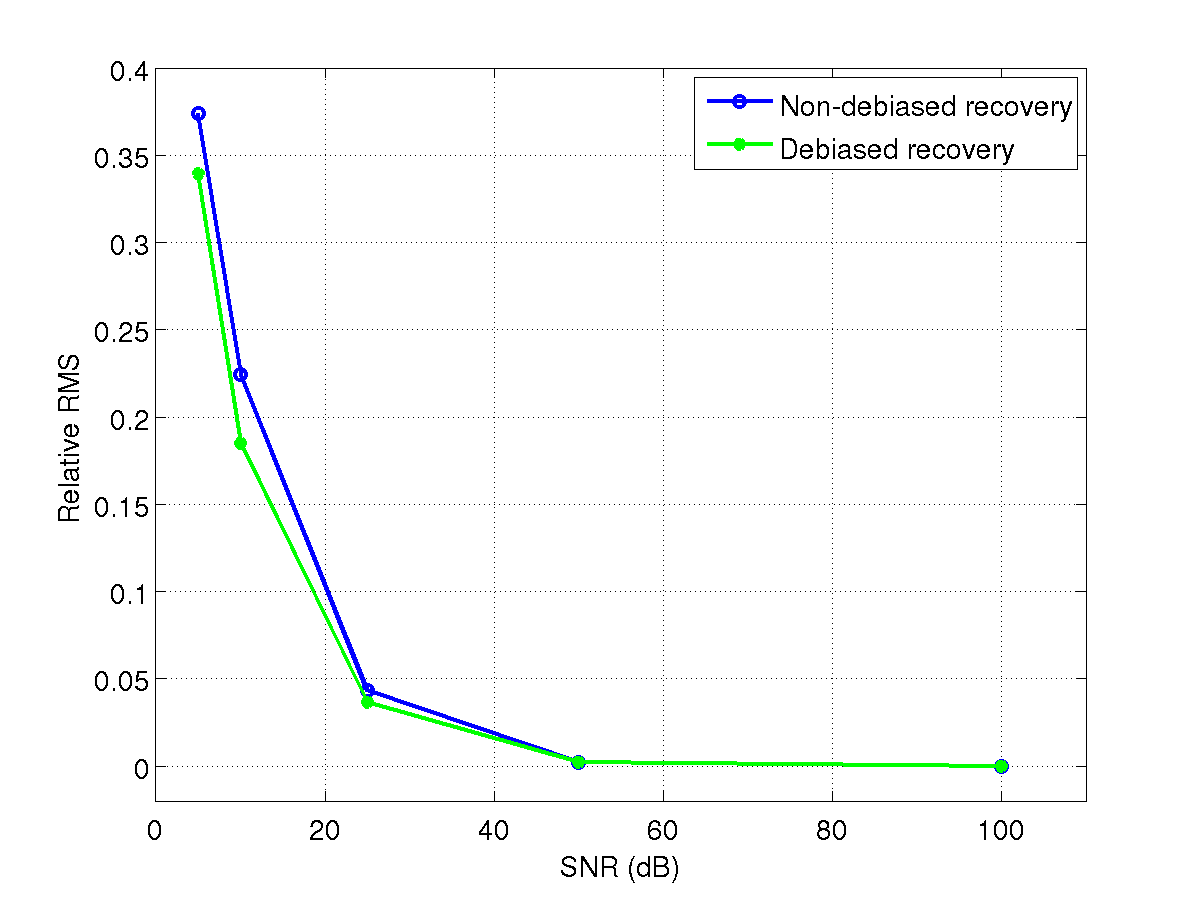}}
\caption{Performance of PhaseLift for Poisson noise.  The stability of
  the algorithm is apparent as its performance degrades gracefully
  with decreasing SNR. (a) Relative MSE on a log-scale for the
  non-debiased recovery. (b) Relative RMS for the original and
  debiased recovery.}
\label{fig:num1}
\end{center}
\end{figure}

\begin{figure}
\begin{center}
\subfigure[]{\includegraphics[width=75mm]{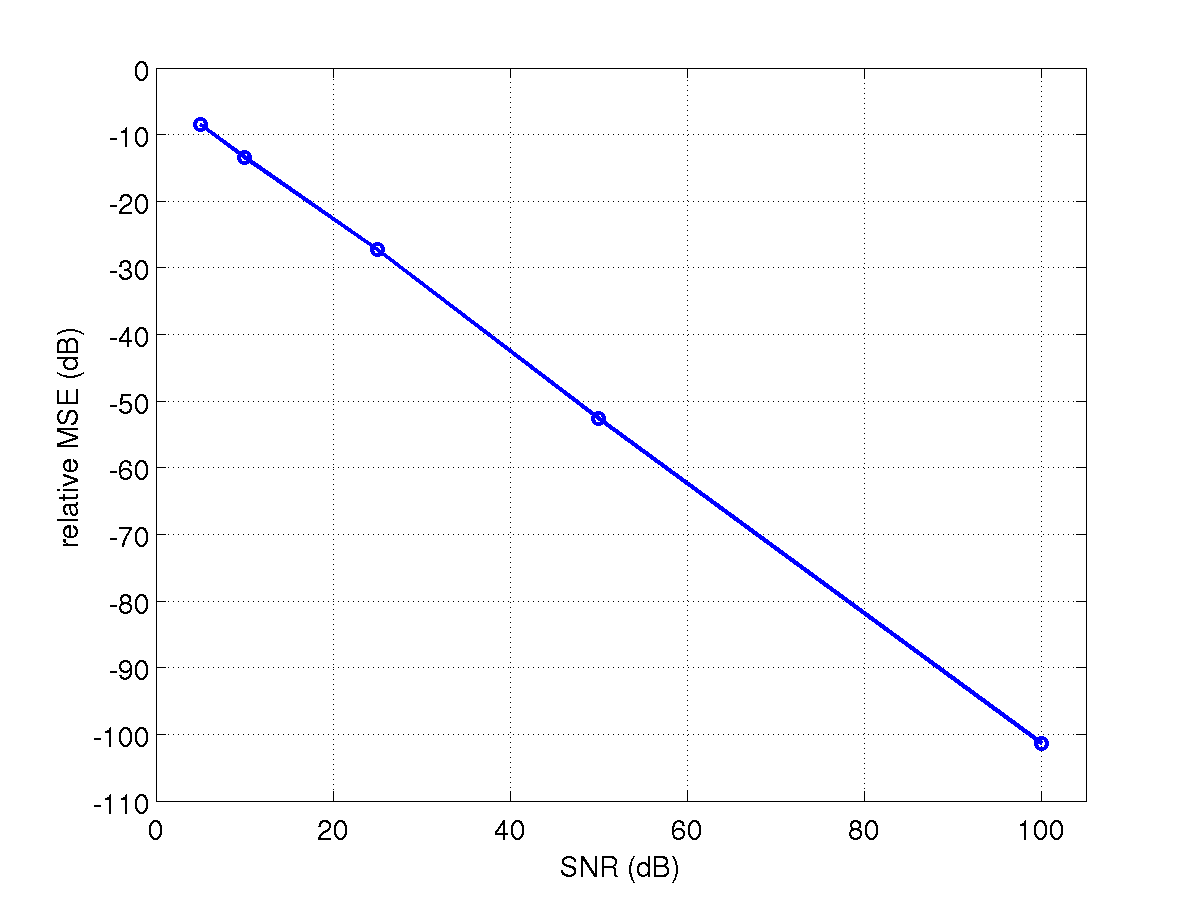}}
\quad
\subfigure[]{\includegraphics[width=75mm]{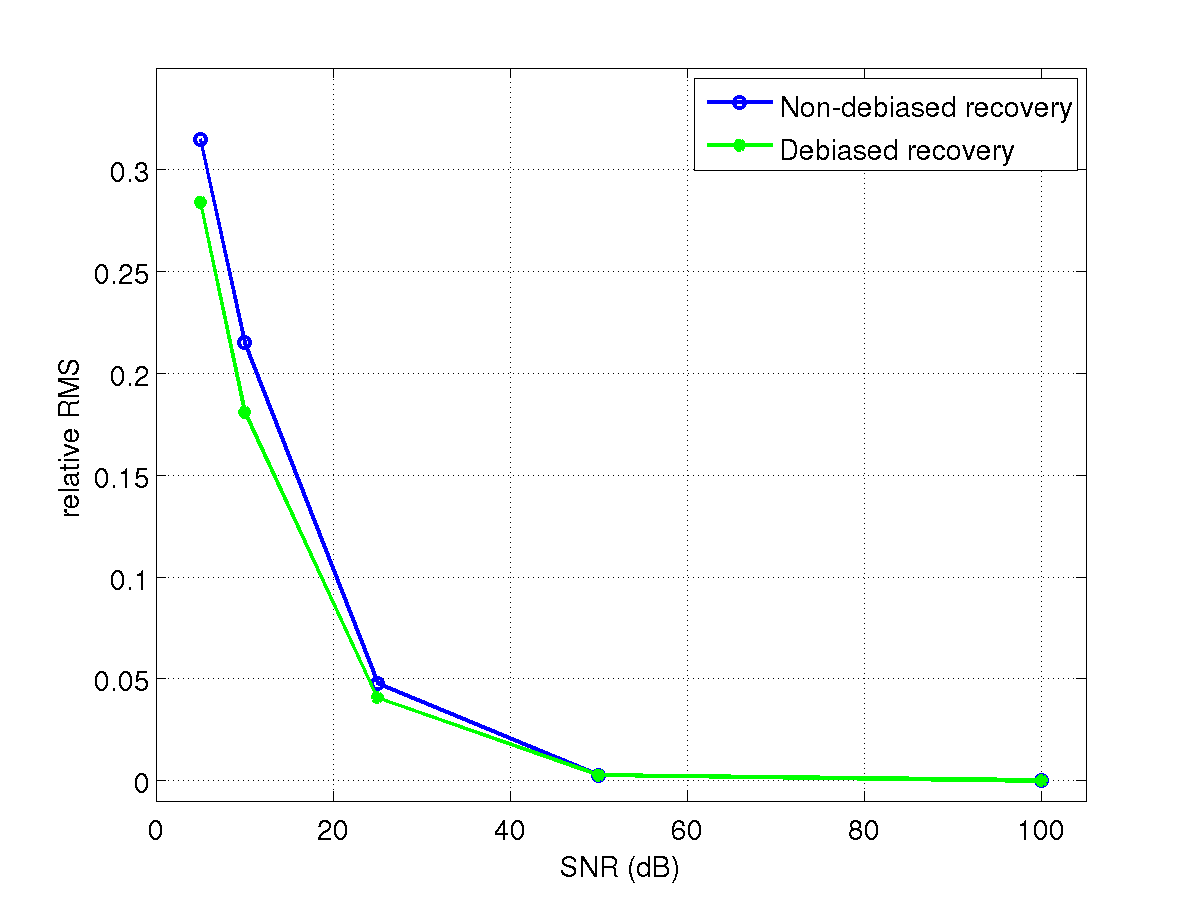}}
\caption{Performance of PhaseLift for Gaussian noise.  (a) Relative
  MSE on a log-scale for the non-debiased recovery. (b) Relative RMS
  for the original and the debiased recovery.}
\label{fig:num1a}
\end{center}
\end{figure}

In the next experiment, we collect Poisson data about a complex-valued
random signal just as above, and work with a fixed SNR set to
15dB. The number of measurements varies so that the oversampling rate
$m/n$ is between 5 and 22 ($m$ is thus between $n \log n$ and $4.5
n\log n$).  We repeat the experiment ten times with different noise
terms and different random measurement vectors for each oversampling
rate; we then record the average relative RMS. Figure~\ref{fig:num2}
shows the average relative RMS of the solution to \eqref{eq:tracemin}
versus the oversampling rate.  We observe that the decrease in the RMS
is inversely proportional to the number of measurements. For instance,
the error reduces by a factor of two when we double the number of
measurements. If instead we hold the standard deviation of the errors
at a constant level, the mean squared error (MSE) reduces by a factor
of about two when we double the number of measurements.
\begin{figure}
\begin{center}
\includegraphics[width=100mm]{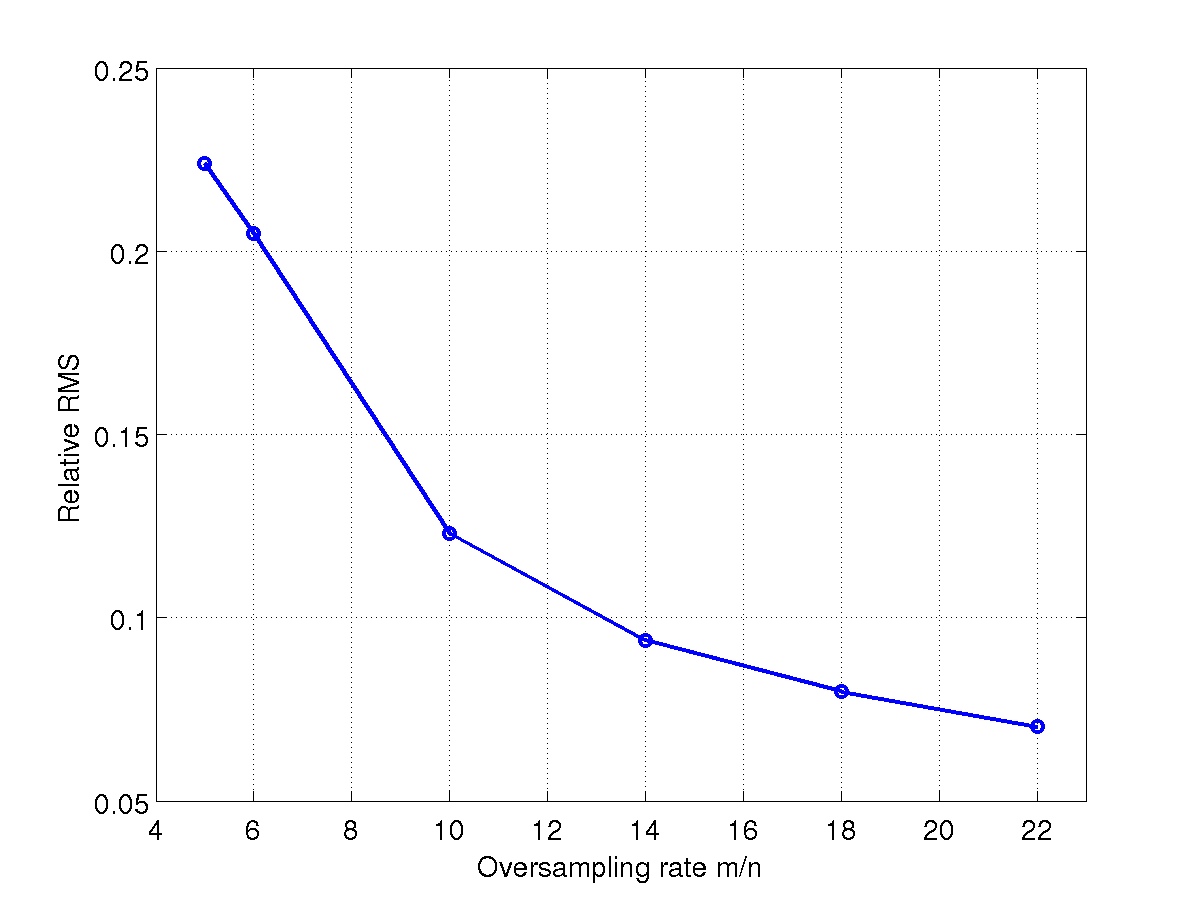}
\caption{Oversampling rate versus relative RMS.}
\label{fig:num2}
\end{center}
\end{figure}

\section{Discussion}
\label{sec:discussion}

In this paper, we have shown that it is possible to recover a signal
exactly (up to a global phase factor) from the knowledge of the
magnitude of its inner products with a family of sensing vectors
$\{\bz_i\}$. The fact that on the order of $n \log n$ magnitude
measurements $|\<\bx, \bz_i\>|^2$ uniquely determine $\bx$ is not
surprising. The part we find unexpected, however, is that what appears
to be a combinatorial problem is solved exactly by a convex
program. Further, we have established the existence of a noise-aware
recovery procedure---also based on a tractable convex program---which
is robust vis a vis additive noise.  To the best of our knowledge,
there are no other results---about the recovery of an arbitrary signal
from noisy quadratic data---of this kind.

An appealing research direction is to study the recovery of a signal
from other types of intensity measurements, and consider other
families of sensing vectors. In particular, {\em structured} random
families would be of great interest.  It also seems plausible that
assuming stochastic errors in Theorem \ref{teo:stability} would allow
to derive sharper error bounds; it would be of interest to know if
this is indeed the case. We leave this to future work.

\begin{small}
\subsection*{Acknowledgements}

E.~C.~is partially supported by NSF via grant CCF-0963835 and the 2006
Waterman Award, and by AFOSR under grant FA9550-09-1-0643.  T.~S.\
acknowledges partial support from the NSF via grants DMS-0811169 and
DTRA-DMS 1042939, and from DARPA via grant N66001-11-1-4090. 
V.~V.~is supported by the
Department of Defense (DoD) through the National Defense Science and
Engineering Graduate Fellowship (NDSEG) Program. E.~C.~thanks Mahdi
Soltanolkotabi for help with Figure \ref{fig:psd}.  V.~V.~acknowledges
fruitful conversations with Bernd Sturmfels.

\end{small}

\begin{small}
\bibliographystyle{abbrv}
\bibliography{phaserefs}
\end{small}

\end{document}